\documentclass[11pt,reqno]{amsart}
\usepackage[margin=1.05in]{geometry}
\usepackage[T1]{fontenc}
\usepackage[utf8]{inputenc}
\usepackage{amsmath,amssymb,booktabs,array,microtype,tikz,needspace}
\usetikzlibrary{arrows.meta,positioning}
\usepackage[hidelinks]{hyperref}
\hypersetup{pdfauthor={Enkai Zhang},pdftitle={The triple rendezvous time of a synchronizing automaton can be floor(4n/3)}}
\newtheorem{lemma}{Lemma}[section]
\newtheorem{theorem}[lemma]{Theorem}
\newtheorem{corollary}[lemma]{Corollary}
\theoremstyle{definition}\newtheorem{definition}[lemma]{Definition}
\theoremstyle{remark}\newtheorem{remark}[lemma]{Remark}
\newtheorem{conjecture}[lemma]{Conjecture}
\theoremstyle{plain}\newtheorem*{theoremA}{Theorem A}

\newcommand{\floor}[1]{\left\lfloor #1\right\rfloor}
\newcommand{\act}{\mathbin{\cdot}}
\title[Triple rendezvous time]{The triple rendezvous time of a synchronizing automaton can be $\lfloor4n/3\rfloor$}
\author{Enkai Zhang}
\address{University of Toronto Scarborough, Toronto, Ontario, Canada}
\email{ek.zhang@mail.utoronto.ca}
\date{}

\providecommand{\Q}{Q}
\providecommand{\Tr}{T_3}
\providecommand{\G}[2]{G_{#1,#2}}
\providecommand{\dep}{d}
\providecommand{\lab}{\lambda}

\providecommand{\CWA}{\operatorname{CWA}}
\providecommand{\MainPrefix}{}
\providecommand{\SuppPrefix}{}
\theoremstyle{plain}
\newtheorem{Stheorem}{Theorem}[section]
\theoremstyle{plain}
\newtheorem{Slemma}[Stheorem]{Lemma}
\theoremstyle{plain}

\theoremstyle{plain}
\newtheorem{Sproposition}[Stheorem]{Proposition}
\theoremstyle{plain}

\theoremstyle{plain}
\newtheorem*{StheoremB}{Theorem B}
\theoremstyle{definition}

\theoremstyle{definition}

\theoremstyle{definition}
\newtheorem{Squestion}[Stheorem]{Question}
\theoremstyle{definition}

\providecommand{\MainPrefix}{}
\providecommand{\SuppPrefix}{}
\hypersetup{hypertexnames=false}
\begin{document}
\hypertarget{article-start}{}
\renewcommand{\SuppPrefix}{Supplement }

\begin{abstract}
For every $n\ge6$, we construct a strongly connected synchronizing automaton
with two input letters and $n$ states whose shortest word merging some three
distinct states has length $\lfloor4n/3\rfloor$. One letter is a permutation;
the other has an image of size $n-1$. We give complete transition maps and
an explicit word attaining the bound. For the lower bound, we assign an
integer to each unordered pair of states and prove that applying either input letter
can decrease it by at most one. For $n\not\equiv0\pmod3$, the same construction
contains a specified state whose shortest merging word, with its partner
chosen freely, has length $\lfloor4n/3\rfloor-1$.
\end{abstract}
\maketitle

\section{Introduction}\label{M-sec:intro}
A complete deterministic automaton consists of a finite state set $Q$ and a
map $Q\to Q$ for each input letter. We use the two letters $a,b$, and write
$q\act w$ for the image of a state under a word, read from left to right.
For $S\subseteq Q$, put $S\act w=\{q\act w:q\in S\}$. A word \emph{merges}
$S$ if $|S\act w|=1$. A word merging $Q$ is a \emph{reset word}, and an
automaton having such a word is \emph{synchronizing}. It is \emph{strongly
connected} if every state is reachable from every other state. The
\emph{rank} of a letter or word is the size of its image on $Q$.

For an automaton with at least three states, Gonze and Jungers
\cite{M-GJ15,M-GJ16} introduced the \emph{triple rendezvous time}
\[
 T_3(A)=\min\{|w|:S\subseteq Q,\ |S|=3,\ |S\act w|=1\}.
\]
Both $S$ and $w$ are chosen in this minimum. Thus $T_3$ concerns the easiest
triple to merge, rather than a prescribed triple or a word resetting the entire state set.

This minimum measures a first obstruction beyond merging a pair: a
non-permutation letter already merges some pair in one step, but creating
a fibre of size three can require a long word. Every reset word merges a
triple, so $T_3$ is a lower bound on its length. More specifically, studying
short words that merge small sets separates local merging mechanisms from
the problem of resetting the entire automaton. Gonze and Jungers relate
this viewpoint to the synchronizing probability function
\cite[Section 3]{M-GJ16}; Behague and Johnson place it in the broader
study of $k$-set rendezvous times \cite{M-BJ22}.

The \v{C}ern\'y automaton has $T_3=n+1$. Gonze and Jungers
constructed binary examples with $T_3=n+3$ for every odd $n\ge9$
\cite[Proposition 4.4]{M-GJ16}. Our value exceeds $n+1$ at every
$n\ge6$, equals $n+3$ at the common odd orders $9,11$, and exceeds
it at every odd $n\ge13$. Gonze and Jungers' upper bound for
strongly connected automata has leading term $(\sqrt5-1)n^2/8$
\cite[Theorem 3.13]{M-GJ16}; the terminal-component reduction
in Supplement Lemma~\ref{S-lem:sink} gives the same leading term
for general synchronizing automata. Whether $T_3$ always has a
linear upper bound remains open \cite[Conjecture 20]{M-BJ22}.

\begin{theoremA}
For every $n\ge6$, there is a strongly connected synchronizing binary
automaton with $n$ states, one permutation letter and one letter of rank
$n-1$, whose triple rendezvous time is exactly $\lfloor4n/3\rfloor$.
\end{theoremA}

Consequently, any upper bound of the form $T_3\le\alpha n+O(1)$ valid
for all synchronizing binary automata must have $\alpha\ge4/3$.
For each displayed automaton, the lower bound applies to every triple,
and a matching word merges one triple. The theorem does not assume that
the family is extremal among all automata of the same size.

Write $n=3r+k$, where $r\ge2$ and $k\in\{0,1,2\}$. Section 3 defines the
automaton $G_{r,k}$ and depicts all three cases. The proof has two parts.
First, an integer label on every state pair gives a lower bound on its merging
distance. After the first collision in a triple, every possible remaining
pair has a sufficiently large label. Second, an explicit word attains that
bound. Sections 4 and 5 carry out these two steps. Section 6 explains why
merging a specified state with a freely chosen partner is a different problem.

The proof uses inverse images and shortest paths on pairs. All steps needed
for Theorem A are given below; Section~3 lists short reset witnesses at the smallest orders.
Separate finite enumerations and restricted-class bounds accompany the article.

\section{Merging pairs of states}\label{M-sec:prelim}
For a nonempty set $S$ of at most two states, define
\[
 d(S)=\min\{|w|:|S\act w|=1\},
\]
with $d(S)=\infty$ if no such word exists. In particular, $d(\{q\})=0$.
For a pair $P$ with finite merging distance, separating the first letter gives
\begin{equation}
 d(P)=1+\min_{x\in\{a,b\}}d(P\act x). \tag{1}\label{M-eq:depth}
\end{equation}
An automaton is synchronizing if and only if every pair has finite merging
distance: one direction follows from a reset word, and the other follows by
repeatedly merging two states in the current image.

For a letter $x$ and a state $z$, let
$x^{-1}(z)=\{q\in Q:q\act x=z\}$. We call this a \emph{collision set} when
it has at least two elements, and a \emph{collision pair} when it has exactly two.

\begin{lemma}[First collision]\label{M-lem:formula}
Let $A$ be a synchronizing binary automaton with at least three states.
If some letter has a collision set of size at
least three, then $T_3(A)=1$. Otherwise,
\[
 T_3(A)=1+\min\{d(\{z,y\}):x\in\{a,b\},\ |x^{-1}(z)|=2,
                         \ y\in Q\act x\setminus\{z\}\}.
\]
\end{lemma}
\begin{proof}
A fiber of size at least three gives a one-letter merging word, and
the empty word merges no distinct states. In the remaining case,
take a word $v=x_1\cdots x_t$ merging a triple $S$, and let $i$ be the first
index for which $|S\act x_1\cdots x_i|<3$. In the absence of a collision set
of size at least three, this first reduction leaves a pair $\{z,y\}$, where
$|x_i^{-1}(z)|=2$ and $y\in Q\act x_i\setminus\{z\}$. The remaining suffix
merges this pair, so $t\ge i+d(\{z,y\})\ge1+d(\{z,y\})$.
Conversely, for any $x,z,y$ in the displayed minimum, take the two states in
$x^{-1}(z)$ and one state in $x^{-1}(y)$. They are distinct. Reading $x$,
then a shortest word merging $\{z,y\}$, merges this triple.
\end{proof}

\begin{lemma}[Integer labels for a lower bound]\label{M-lem:potential}
Assign an integer $\lambda(P)$ to every pair. Suppose that
$\lambda(P)\le1$ whenever $d(P)=1$, and
\[
 \lambda(P)\le 1+\lambda(P\act x)
\]
for every pair with $d(P)\ge2$ and every letter $x$. Then
$d(P)\ge\lambda(P)$ for every pair.
\end{lemma}
\begin{proof}
For finite $d(P)$, use induction in \eqref{M-eq:depth}, choosing a first
letter of a shortest merging word. The base case is the assumed inequality
at distance one. Pairs with infinite merging distance satisfy the conclusion
trivially.
\end{proof}

\section{The construction}\label{M-sec:family}
The design uses a long cycle and a cycle of half its length. Rotation
preserves their relative position modulo the shorter length, while a few
modified transitions provide a single collision. The collision letter
also omits one state from its image, excluding that state as the third
state's position immediately after a collision. In the base construction,
this omitted state is precisely a cheaply merging partner of the collision
image. A short added chain supplies the other state counts. The following
maps implement this design; the pair labels in Section~\ref{M-sec:lower}
control all alternative routes.

Fix $r\ge2$ and $k\in\{0,1,2\}$. The states are the disjoint sets
\[
 C=\{0,1,\ldots,2r-1\},\qquad
 D=\{0',1',\ldots,(r-1)'\},\qquad E=\{e_1,\ldots,e_k\}.
\]
Here $E$ is empty for $k=0$. A primed symbol is the name of a state in $D$;
it is never an alternative name for an unprimed state. Define
\[
 b(i)=(i+1)\bmod 2r,\qquad
 b(j')=((j+1)\bmod r)',\qquad b(e_j)=e_j.
\]
Set $a(x)=b(x)$ except at the states in the following table. A dash means
that the state is absent.
\begin{equation}
\begin{array}{c|ccc}
 x & a(x),\ k=0 & a(x),\ k=1 & a(x),\ k=2\\\hline
 r-2       &0&e_1&e_1\\
 (r-2)'    &0&e_1&e_1\\
 2r-2      &r-1&r-1&r-1\\
 2r-1      &(r-1)'&(r-1)'&(r-1)'\\
 (r-1)'    &2r-1&2r-1&2r-1\\
 e_1       &-&0&e_2\\
 e_2       &-&-&0
\end{array}\tag{2}\label{M-eq:a}
\end{equation}
These rules specify both letters at every state for every permitted parameter.
For $r\ge3$, Figure~\ref{M-fig:construction} gives an equivalent description by the two
functional digraphs, drawn separately to avoid overlapping letter labels.

For the proofs, abbreviate $c=r-2$ and put $z=0$ for $k=0$, and $z=e_1$
for $k\ge1$. The only collision pair of $a$ is $\{c,c'\}$, with image $z$;
its only omitted image is $0'$. Hence
\[
a^{-1}(z)=\{c,c'\},\qquad Q\act a=Q\setminus\{0'\},
\qquad \operatorname{rank}(a)=3r+k-1.
\]

The length scale can be seen in one explicit route. The pair $\{0,1\}$
follows
\[
 \{0,1\}\xrightarrow{b^{2r-2}}\{2r-2,2r-1\}
 \xrightarrow{a}\{r-1,(r-1)'\}
 \xrightarrow{b^{2r-1}}\{c,c'\}
 \xrightarrow{a}\{z\}.
\]
The two long rotations account for $4r-3$ letters of this route. The
modified transition at $2r-2$ sends the first pair into corresponding
positions on the two cycles, from which the second rotation reaches the
collision pair. Section~\ref{M-sec:upper} supplies the short prefix that
brings an initial triple to $\{0,1\}$. This route proves attainability;
the separate label argument is needed to exclude shorter ways of merging
any triple.

\newcommand{\constructionpanel}[1]{%
\begin{tikzpicture}[x=1cm,y=1cm,>=Stealth,font=\small,
 state/.style={draw,rounded corners=1pt,minimum width=9mm,minimum height=6mm,inner sep=2pt},
 every path/.style={line width=.5pt}]
\path[use as bounding box] (-.2,-.5) rectangle (14.7,4.15);
\node[anchor=west,font=\bfseries] at (0,3.96) {$k=#1$};
\node[anchor=west] at (0,3.5) {letter $b$};
\node[anchor=west] at (5.1,3.5) {letter $a$};
\draw[densely dotted] (4.55,-.2)--(4.55,3.6);
\node[state] (b0) at (.45,2.65) {$0$};
\node[state] (bp) at (3.55,2.65) {$2r-1$};
\draw[->] (b0)--node[above] {$\cdots$} (bp);
\draw[->] (bp.south) to[bend left=30] (b0.south);
\node[state] (bd0) at (.45,1.25) {$0'$};
\node[state] (bdp) at (3.55,1.25) {$(r-1)'$};
\draw[->] (bd0)--node[above] {$\cdots$} (bdp);
\draw[->] (bdp.south) to[bend left=30] (bd0.south);
\ifnum#1>0
\node[state] (be1) at (1.1,-.02) {$e_1$};
\draw[->] (be1) edge[loop left,looseness=5] (be1);
\fi
\ifnum#1>1
\node[state] (be2) at (3,-.02) {$e_2$};
\draw[->] (be2) edge[loop right,looseness=5] (be2);
\fi
\node[state] (a0) at (5.7,2.65) {$0$};
\node[state] (ac) at (8.5,2.65) {$r-2$};
\draw[->] (a0)--node[above] {$\cdots$} (ac);
\node[state] (ad0) at (5.7,1.35) {$0'$};
\node[state] (adc) at (8.5,1.35) {$(r-2)'$};
\draw[->] (ad0)--node[above] {$\cdots$} (adc);
\ifnum#1=0
\draw[->] (ac.north) to[bend right=31] (a0.north);
\draw[->] (adc.north) .. controls (8.5,2.1) and (5.7,1.9) .. (a0.south);
\else
\node[state] (ae1) at (10.6,2.65) {$e_1$};
\draw[->] (ac)--(ae1);
\draw[->] (adc.east) to[out=0,in=-90] (ae1.south);
\ifnum#1=1
\draw[->] (ae1.north) to[bend right=21] (a0.north);
\else
\node[state] (ae2) at (12.5,2.65) {$e_2$};
\draw[->] (ae1)--(ae2);
\draw[->] (ae2.north) to[bend right=15] (a0.north);
\fi
\fi
\node[state] (am0) at (5.7,.05) {$r-1$};
\node[state] (am1) at (8.5,.05) {$2r-2$};
\draw[->] (am0)--node[above] {$\cdots$} (am1);
\draw[->] (am1.south) to[bend left=26] (am0.south);
\node[state] (az) at (11.15,.15) {$2r-1$};
\node[state] (adz) at (13.45,.15) {$(r-1)'$};
\draw[->] (az.north east) to[bend left=17] (adz.north west);
\draw[->] (adz.south west) to[bend left=17] (az.south east);
\end{tikzpicture}}

\begin{figure}[!htbp]
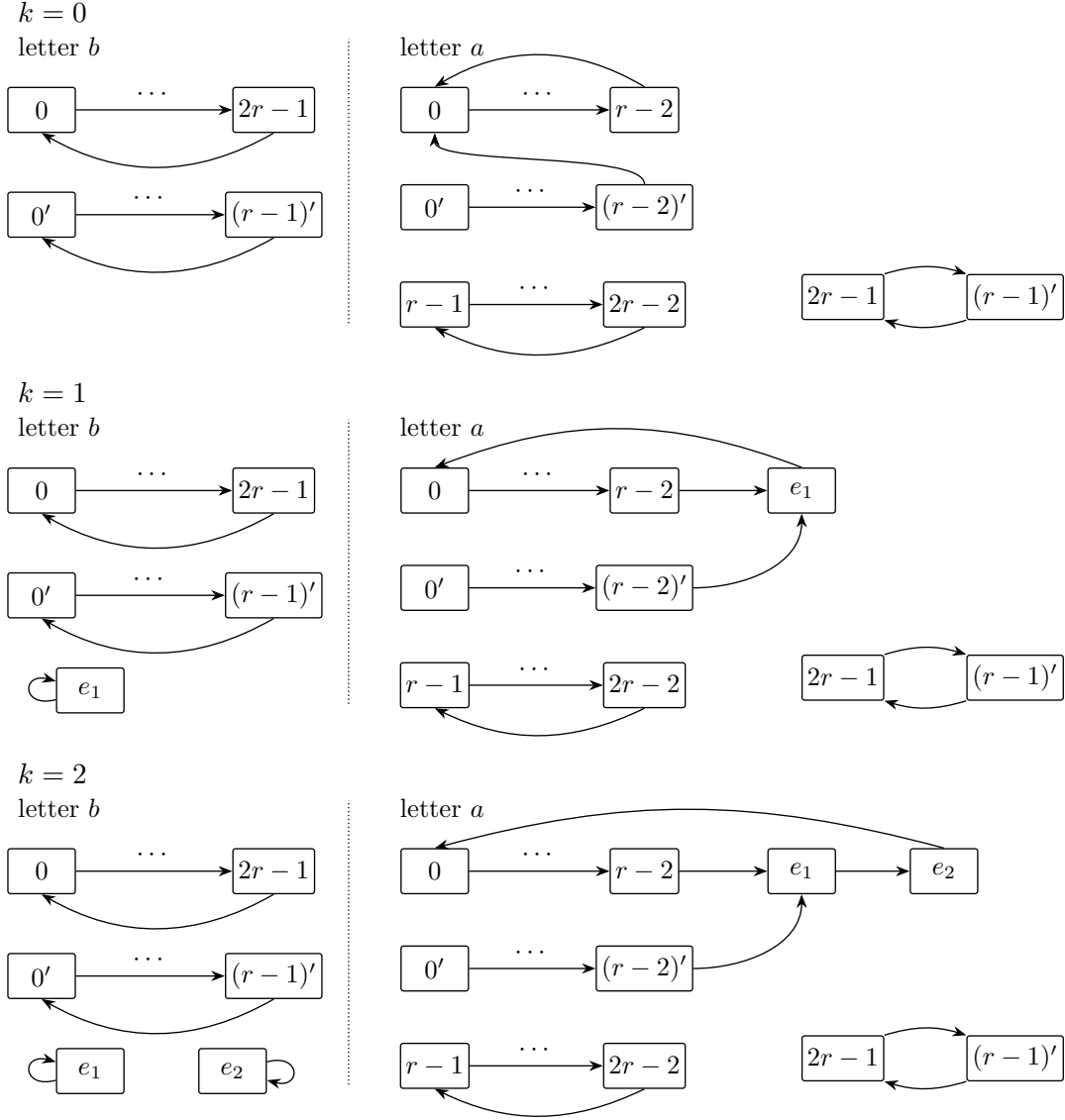

\centering
\constructionpanel{0}\par\vspace{4mm}
\constructionpanel{1}\par\vspace{4mm}
\constructionpanel{2}
\caption{Complete transition maps of $G_{r,k}$ for every $r\ge3$.
The left and right columns draw the actions of $b$ and $a$ separately on the
same states. An arrow marked $\cdots$ passes through every consecutive index
between its endpoints, in order: for example, the arrow from $0$ to $r-2$ on the right is the path
$0\to1\to\cdots\to r-2$. An empty list, as occurs when $r=3$, means a
single edge. Unannotated arrows are single transitions. Each arrow in a
column carries that column's letter. After expanding the indicated paths, every state has exactly one
outgoing edge under each letter; there are no additional transitions.}
\label{M-fig:construction}
\end{figure}

\clearpage

\Needspace{6\baselineskip}
\begin{lemma}\label{M-lem:sc}
The automaton $G_{r,k}$ is strongly connected.
\end{lemma}
\begin{proof}
The letter $b$ cycles through $C$ and through $D$. The two $a$-transitions
$2r-1\leftrightarrow(r-1)'$ connect these cycles in both directions.
Every added state $e_j$ lies on the $a$-cycle containing $0$.
\end{proof}

For a pair with one state in each of $C,D$, write it as $\{u,v'\}$ and define
its \emph{offset} to be $\delta=(v-u)\bmod r$. The letter $b$ preserves the
offset, and all pairs of a given offset form a single $b$-orbit. In particular,
the offset-zero pairs are
\[
 O_j=\{j,(j\bmod r)'\},\quad 0\le j<2r,
 \qquad O_j\act b=O_{j+1\bmod 2r}.
\]
Thus $O_c$ is the collision pair. Every $O_j$ merges under
$b^{(c-j)\bmod 2r}a$.

\paragraph{The three automata at $r=2$.}
Here $C=\{0,1,2,3\}$, $D=\{0',1'\}$ and $c=0$.
Each of the following words sends the entire state set to the state
in the last column, by substitution in \eqref{M-eq:a}.
\[
\begin{array}{ccl}
k&\text{reset word}&\text{image}\\
0&\texttt{ababbabababbbaabbabbba}&0\\
1&\texttt{ababbabbbabbabbbabbabaabbabbba}&e_1\\
2&\texttt{ababbabababbabababbabababbbabaabbabbba}&e_1
\end{array}
\]
\begin{lemma}\label{M-lem:sync}
The automaton $G_{r,k}$ is synchronizing.
\end{lemma}
\begin{proof}
The preceding reset words settle $r=2$. For $r=3$, direct substitution in \eqref{M-eq:a} verifies the reset words
in Table~\ref{M-tab:small-sync}. It remains to treat $r\ge4$, so $c\ge2$. We show
that every pair reaches an offset-zero pair or merges.

\smallskip\noindent\emph{Pairs contained in $C$ or $D$.}
For $\{i,j\}\subseteq C$, $i<j$, rotate by $b^{2r-1-j}$ to obtain
$\{u,2r-1\}$, where $u=i+2r-1-j\in[0,2r-2]$. The letter $a$ sends this
to a pair with one state in $D$ and one in $C$, except possibly when
$u=c$ and $k\ge1$, in which case it gives $\{e_1,(r-1)'\}$.
When $u=2r-2$, the image is already $O_{c+1}$.
For a pair in $D$, rotate one state to $(r-1)'$ and apply $a$. The resulting
pair is mixed, contained in $C$, or contains $e_1$.

\smallskip\noindent\emph{Pairs containing added states.}
Apply $a^{k+1-j}$ to a pair containing $e_j$. One image is $0$. If the other
is in $C$ or $D$, the pair is of a type just considered. If the other image
is an added state $e_i$, apply $a^{k+1-i}$ again; the images are $0$ and
$k+1-i\le2\le c$, hence lie in $C$.
To finish the reduction, first eliminate pairs in $D$, then those with
an added state, then those in $C$. The only case in which the final step reintroduces an added state gives $\{e_1,(r-1)'\}$.
Applying $a^k$ gives either $\{0,2r-1\}$ or $\{0,(r-1)'\}$.
The latter is mixed; the former reaches $O_{c+1}$ under $b^{2r-1}a$.
Consequently every pair reaches a mixed pair or merges.

\smallskip\noindent\emph{Mixed pairs when $k\ge1$.}
A nonzero-offset pair can be rotated to $\{c,v'\}$ with $v\ne c$.
If $v=r-1$, applying $a$ gives $\{e_1,2r-1\}$. For $k=2$, the next
$a^2$ gives $\{0,2r-1\}$ and the preceding reduction finishes. For $k=1$,
the next $a$ gives $\{0,(r-1)'\}$, of offset $r-1$. After rotation, this
falls in the other case below, since its primed index is $r-3$.
Otherwise the first $a$ gives $\{e_1,(v+1)'\}$. Since $b$ fixes $e_1$,
rotate the second entry to $c'$ and apply $a$. The result is
$\{0,e_1\}$ for $k=1$ or $\{e_2,e_1\}$ for $k=2$.
One or two further $a$-steps give $\{0,1\}$, which reaches $O_{c+1}$ under
$b^{2r-2}a$.

\smallskip\noindent\emph{Mixed pairs when $k=0$.}
All offsets in this paragraph are taken modulo $r$. Offset zero has already
been settled. For offset one, rotate to $\{c,(r-1)'\}$ and apply $a$ to
reach $\{0,2r-1\}$.
For the other offsets we verify two reductions:
\[
 \delta\longmapsto-\delta-1\quad(\delta\notin\{0,1\}),
 \qquad
 \delta\longmapsto-\delta\quad(\delta\notin\{0,r-1\}).
\]
For the first, rotate to $\{c,v'\}$ and apply $a$, obtaining
$\{0,(v+1)'\}$. Choose $s\in[0,r-1]$ with
$s\equiv c-(v+1)\equiv-1-\delta$. Then $s\ne c$.
The word $b^s a$ gives $\{0,d\}$, where $d=s+1\in[1,r]$ and
$d\equiv-\delta$. If $d=1$, this pair reaches offset zero as above.
Otherwise the reduction for pairs in $C$ gives
$\{2r-d,(r-1)'\}$, of offset $d-1\equiv-\delta-1$.

For the second reduction, rotate to $\{2r-1,(\delta-1)'\}$ and apply $a$.
The result is $\{(r-1)',\delta'\}\subseteq D$. Rotating by
$b^{r-1-\delta}$ and applying $a$ gives
$\{(-\delta-1)',2r-1\}$, of offset $-\delta$.
The excluded offsets are precisely those for which one of the exceptional
transitions would change this calculation.
Offset $r-1$ reaches zero by the first reduction. For
$2\le\delta\le r-2$, perform the second reduction and then the first;
the offset becomes $\delta-1$. Repetition reaches offset one, which has
already been settled. This completes the proof that every pair merges.
\end{proof}

\begin{table}[htbp]
\centering\small
\begin{tabular}{cclc}
\toprule
$k$ & Length & Reset word for $G_{3,k}$ & Image\\\midrule
0 &37&\texttt{aabbabababaabbaababababbaabbabababbaa}&$0$\\
1 &46&\texttt{aabababaabbababbbaabaaabbabababbaaabbabababbaa}&$e_1$\\
2 &54&\texttt{aabbabababaabbbaabbabababbaaaabbabababbaaaabbabababbaa}&$e_1$\\
\bottomrule
\end{tabular}
\caption{Reset words at $r=3$.
Their reset property follows by applying the maps in \eqref{M-eq:a} to
every state. Their optimality is not needed.}
\label{M-tab:small-sync}
\end{table}

\section{A lower bound on every triple}\label{M-sec:lower}
Let $r\ge2$ and $k\in\{0,1,2\}$.
By Lemma~\ref{M-lem:formula}, it is enough to prove
$d(\{z,y\})\ge4r+k-1$ for every $y\in Q\setminus\{z,0'\}$.
We assign labels to all pairs and verify the one-step inequalities of
Lemma~\ref{M-lem:potential}. On the pairs $O_j$, the labels count the letters
in the direct merging word $b^{(c-j)\bmod 2r}a$; the remaining labels control
transitions outside this orbit.

Define a height on the unprimed and added states by
\[
 h(i)=i\quad(i\in C),\qquad h(e_j)=j-k-1\quad(1\le j\le k).
\]
No height is assigned to states in $D$. For a pair $P\not\subseteq D$, put
\[
 \mu(P)=\min\{h(x):x\in P\cap(C\cup E)\}.
\]
Assign the following integer labels to pairs:
\begin{equation}
\lambda(P)=
\begin{cases}
 r-1-j,&P=O_j,\quad 0\le j\le c,\\
 3r-1-j,&P=O_j,\quad c+1\le j\le 2r-1,\\
 2r-1,&P\subseteq D,\\
 4r-1-\mu(P),&\text{otherwise}.
\end{cases}\tag{3}\label{M-eq:lambda}
\end{equation}
In particular, $\lambda(O_c)=1$. Every pair containing $z$ that can arise
after a first collision has label $4r-1+k$: the one exception when $k=0$
would be $O_0=\{0,0'\}$, but $0'\notin Q\act a$.

\begin{lemma}\label{M-lem:lip}
For every pair $P\ne O_c$ and every $x\in\{a,b\}$, the image $P\act x$
is a pair and
\[
 \lambda(P)\le1+\lambda(P\act x).
\]
\end{lemma}
\begin{proof}
Only $a$ merges a pair, and its unique collision pair is $O_c$.
We check the three types in \eqref{M-eq:lambda}. The exceptional
long-cycle indices $c,c+1,2r-2,2r-1$ are distinct for every
$r\ge2$; at $r=2$ they are $0,1,2,3$. The labels on
$O_0,O_1,O_2,O_3$ are then $1,4,3,2$, respectively.

\smallskip\noindent\emph{1. Offset-zero pairs.}
Under $b$, $O_j$ goes to $O_{j+1\bmod2r}$, and its label decreases by one
for $j\ne c$, including the wrap from $2r-1$ to $0$.
Under $a$, the same happens unless $j\in\{2r-2,c+1,2r-1\}$.
For $j=2r-2$, the image is $\{r-1,z\}$ and has label at least $4r-1$.
For $j=2r-1$, the pair is fixed. For $j=c+1$, the image is
$\{c+2,2r-1\}\subseteq C$, with label $3r-1\ge2r-1=\lambda(O_{c+1})-1$.

\smallskip\noindent\emph{2. Pairs in $D$.}
The letter $b$ preserves $D$ and the label. So does $a$ on a pair avoiding
$c'$ and $(r-1)'$. If the pair contains $c'$, its image contains $z$ and
cannot be $O_0$, since $0'$ has no $a$-preimage; its label is at least $4r-1$.
Otherwise, a pair $\{(r-1)',j'\}$ with $j\notin\{c,r-1\}$ goes to
$\{2r-1,(j+1)'\}$. This is not an offset-zero pair, and has label $2r$.
Each case satisfies the required inequality.

\smallskip\noindent\emph{3. All other pairs.}
Choose $s\in P\cap(C\cup E)$ with $h(s)=\mu(P)$, and let $t$ be the
other state. Under $b$, the image is neither contained in $D$ nor offset
zero. Moreover $h(s\act b)\le h(s)+1$. Hence
\[
 \lambda(P\act b)\ge4r-1-h(s\act b)\ge\lambda(P)-1.
\]
For $a$, there are five possibilities.

If $s=2r-1$, then $t$ must lie in $D$: a distinct state of $C\cup E$ would
have smaller height, contradicting the choice of $s$. Write $t=j'$.
Since $P\ne O_{2r-1}$, $j\ne r-1$, and $\lambda(P)=2r$.
For $j=c$, the image contains $z$ and has label at least $4r-1$.
Otherwise it lies in $D$ and has label $2r-1$.

If $s\in C\setminus\{c,2r-2,2r-1\}$, then $s\act a=s+1$.
The image cannot be contained in $D$. Nor can it be $O_{s+1}$:
a primed preimage of its other entry would make $P=O_s$, while the
unprimed preimage $2r-1$ would require $s\in\{c,2r-2\}$.
Therefore its label is at least $4r-1-(h(s)+1)=\lambda(P)-1$.

If $s=c$, the image contains $z$ and cannot be $O_0$. Its label is at
least $4r-1$, as required.

If $s=2r-2$ and $t=2r-1$, the image is $O_{c+1}$ and
$\lambda(P)=2r+1=1+\lambda(O_{c+1})$. Otherwise the image contains
$r-1$, is not offset zero, and has label at least $3r$.

Finally, if $s=e_j$, its image has height $h(s)+1$. The image pair is not
in $D$ and cannot be $O_0$, again because $0'$ has no $a$-preimage. Its
label is at least $\lambda(P)-1$. These possibilities exhaust $C\cup E$.
\end{proof}

\begin{theorem}\label{M-thm:lower}
For $r\ge2$ and $k\in\{0,1,2\}$, $T_3(G_{r,k})\ge4r+k$.
\end{theorem}
\begin{proof}
Lemmas~\ref{M-lem:potential} and \ref{M-lem:lip} give $d(P)\ge\lambda(P)$.
The only collision is $a^{-1}(z)=O_c$, and $Q\act a=Q\setminus\{0'\}$.
The first-collision formula therefore gives
\[
 T_3(G_{r,k})
 =1+\min_{y\in Q\setminus\{z,0'\}}d(\{z,y\})
 \ge1+(4r-1+k)=4r+k.
\]
\end{proof}

\section{A word attaining the bound}\label{M-sec:upper}
\begin{theorem}\label{M-thm:upper}
The word
\[
 a^{k+1}b^{2r-2}ab^{2r-1}a
\]
has length $4r+k$. At $r=2$ it merges $\{0,2,0'\}$ for $k=0,2$
and $\{0,1,0'\}$ for $k=1$. At $r\ge3$ it merges the triple $\{c,c',0\}$ when $k=0$, and
$\{c,c',e_1\}$ when $k\ge1$.
\end{theorem}
\begin{proof}
The two collision states $c,c'$ reach $0$ after $a^{k+1}$.
At $r=2$, the third states $2,1,2$ for $k=0,1,2$ reach $1$
under $a,a^2,a^3$, respectively, since $a$ interchanges $1$ and $2$.
At $r\ge3$, $a(0)=1$ and $e_1\act a^k=0$ when $k>0$.
Thus in every case the resulting pair is $\{0,1\}$. It follows
the route displayed in Section~\ref{M-sec:family}, where the mixed pair is
$O_{c+1}$ and the collision pair is $O_c$.
The long final rotation uses $c+1+(2r-1)\equiv c\pmod{2r}$.
The total length is $(k+1)+(2r-2)+1+(2r-1)+1=4r+k$.
\end{proof}

\begin{proof}[Proof of Theorem A]
Lemmas~\ref{M-lem:sc} and \ref{M-lem:sync} give strong connectivity and
synchronization. The maps have the stated ranks by construction.
Theorems~\ref{M-thm:lower} and \ref{M-thm:upper} give
$T_3(G_{r,k})=4r+k=\lfloor4(3r+k)/3\rfloor$.
Every $n\ge6$ has a unique representation $n=3r+k$ with $r\ge2$ and
$k\in\{0,1,2\}$.
\end{proof}

\section{Merging a specified state}\label{M-sec:cwa-main}
\begin{definition}\label{M-def:cwa}
For a state $q$ of an automaton $A$, its \emph{compress-with-another
threshold} is
\[
 \tau_A(q)=\min\{|w|:q\act w=p\act w\text{ for some }p\in Q\setminus\{q\}\},
\]
with value $\infty$ if no such word exists. The partner is chosen freely.
\end{definition}

\begin{remark}\label{M-rem:cwa}
In $G_{r,0}$ the pair $O_0=\{0,0'\}$ has label $r-1$ and
merges under $b^{r-2}a$. Every other pair containing $0$ has
label $4r-1$, so $\tau(0)=r-1$ for $r\ge2$. The only cheap partner,
$0'$, has no $a$-preimage; hence $O_0$ cannot be the surviving pair
after the first collision of a triple. In particular, $T_3=4r$.

The largest specified-state threshold in $G_{r,0}$ is $2r$.
Every state lies in a pair $O_i$ with a merging word of length at
most $2r$. For the unprimed state $r-1$, the pair $O_{r-1}$ has
label $2r$ and merges under $b^{2r-1}a$. All other partners
give pairs with $\mu\le r-1$ and therefore with label at least
$4r-1-(r-1)=3r$. Thus $\tau(r-1)=2r$, including at $r=2$.
\end{remark}

\begin{corollary}\label{M-cor:tau}
For $k\in\{1,2\}$ and $1\le j\le k$,
\[
 \tau_{G_{r,k}}(e_j)=4r+k-j.
\]
In particular, $\tau(e_1)=\lfloor4n/3\rfloor-1$. Every state in $C\cup D$ has
threshold at most $2r$.
\end{corollary}
\begin{proof}
Any pair containing $e_j$ is covered by the last case of \eqref{M-eq:lambda},
and $\mu(P)\le h(e_j)$. Its merging distance is therefore at least
$4r-1-h(e_j)=4r+k-j$.
For the matching upper bound, first reach $\{0,1\}$ by
$a^{k+1-j}$. When $r\ge3$, set $e_{k+1}=0$ and choose
$e_{j+1}$ as the partner. The two states reach $0$ and $1$
under this prefix because $a(0)=1$.
When $r=2$, choose the unprimed partner $2$ for $(k,j)=(1,1)$
or $(2,2)$, and $1$ for $(k,j)=(2,1)$. In the first two
cases the prefix is $a$, with $a(e_k)=0$ and $a(2)=1$;
in the last case it is $a^2$, with $e_1\act a^2=0$ and
$1\act a^2=1$. In all cases the suffix
$b^{2r-2}ab^{2r-1}a$ merges $\{0,1\}$ in $4r-1$ letters.
The total length is $(k+1-j)+(4r-1)=4r+k-j$.

Finally every unprimed state belongs to an $O_i$, which merges under
$b^{(c-i)\bmod2r}a$ in at most $2r$ letters. Every primed state belongs
to two such pairs, so the same bound covers it as well.
\end{proof}

Szyku\l{}a asks whether the largest specified-state threshold in a
strongly connected synchronizing automaton is $O(n)$
\cite[Open Problem 4]{M-Szykula2018}, \cite[Conjecture 6]{M-Szykula2026}.
Corollary~\ref{M-cor:tau} supplies a lower bound of
$\lfloor4n/3\rfloor-1$ when $n\equiv1,2\pmod3$.
In the separate family of D\.{z}yga \cite{M-Dzyga2018}, the predicted
threshold at $n=3p$ is $4p$ \cite[Conjecture 8]{M-Szykula2026}.
Wang and Dong prove it for $p\ge3$ \cite{M-WD};
a path-automaton construction gives a state of threshold
$\lfloor4n/3\rfloor$ for every $n\ge9$ \cite[Theorem C]{M-Path}.
The specified-state results in the present article concern $G_{r,k}$.

\section{Discussion}\label{M-sec:conj}
Theorem A provides a lower construction for every $n\ge6$. It does not bound
$T_3$ from above on arbitrary synchronizing automata.
The following binary conjecture records a pattern in the finite data;
its computational scope is listed immediately below.
\begin{conjecture}\label{M-conj:main}
For every $n\ge6$, the maximum triple rendezvous time of a synchronizing
$n$-state automaton with two letters is $\lfloor4n/3\rfloor$.
\end{conjecture}
\begin{center}\small
\begin{tabular}{@{}>{\raggedright\arraybackslash}p{6.3cm}cl@{}}
\toprule Computed class & State count & Maximum $T_3$\\\midrule
All synchronizing binary automata & $3\le n\le10$
 & $4,5,7,8,9,10,12,13$\\
Permutation letter and rank-$(n-1)$ letter, all cycle types
 & $6\le n\le12$ & $8,9,10,12,13,14,16$\\
The same class, 18 specified cycle types
 & $n=13$ & $17$ in those types\\\bottomrule
\end{tabular}
\end{center}
Canonical enumeration covers $n\le9$; exhaustive branch-and-bound handles
$n=10$. Supplement Section S1 gives the methods, reductions and the
18 selected cycle types at $n=13$, where the enumeration is partial.
The exception $T_3=7$ at $n=5$ explains the conjecture's starting range.
These finite results leave open both a matching upper bound and families
with a larger asymptotic coefficient.

\paragraph{Data and code.}
The ancillary program verifies the three $r=2$ automata from their transition maps.
The \hyperlink{supplement-start}{additional-results document} contains
the finite enumerations and restricted-class bounds, with a table directing
the reader to their full hypotheses and proofs.

\paragraph{Use of generative AI}
Generative AI tools, including OpenAI Codex, assisted with mathematical
analysis, research programming, literature review and manuscript preparation.
The author takes responsibility for the contents.

\enlargethispage{3\baselineskip}

\paragraph{Electronic evidence.}
The computational data and verification programs accompany the preprint at
\url{https://doi.org/10.5281/zenodo.22561629}.

\clearpage
\hypertarget{supplement-start}{}
\setcounter{section}{0}\setcounter{subsection}{0}
\setcounter{equation}{0}
\setcounter{table}{0}\setcounter{figure}{0}
\renewcommand{\MainPrefix}{Main }\renewcommand{\SuppPrefix}{}
\section*{Additional results and computational data}

This document gives finite enumeration results and upper bounds for restricted
classes of synchronizing automata, together with their proofs and computation
details. These are additional results: the companion article contains the
complete proof of its construction theorem, including the small
synchronization cases.

\raggedbottom

\section*{Notation and contents}
The state action, rank, merging distance $d(P)$ and triple rendezvous time
$T_3$ are defined in the main article. An \emph{involution} is a
permutation $b$ with $b^2=\mathrm{id}$; its nontrivial disjoint cycles are
transpositions, each exchanging exactly two states. An automaton is
\emph{circular} if one letter permutes all its states in a single cycle.
The \emph{functional digraph} of $a$ has the edges $q\to a(q)$.
A vertex on one of its directed cycles is called \emph{cyclic}.
For a map of rank $n-1$, the noncyclic vertices form a single directed
path into a cycle. We call that path the \emph{tail}; the notation for
its vertices and cycle positions is introduced before the relevant proofs.
Abbreviations TT, TC and RR distinguish transpositions whose endpoints
are both on the tail, one on the tail and one on a cycle, or both on
cycles, respectively. They are local notation, not additional classes
of letters.

For comparison, $\CWA(n)$ denotes the maximum of $\tau_{\mathcal A}(q)$
over all states of all strongly connected synchronizing binary automata
with $n$ states. The specified-state threshold $\tau$ is defined in
Main Definition~\ref{M-def:cwa}.

Sections S1--S9 and Appendix A belong to this additional-results document. References
prefixed by ``Main'' point to the main article; all others are local.
The finite maxima are Theorem B in Section~\ref{S-sec:comp}.
Section~\ref{S-sec:circ} is background on a circular permutation letter,
not an additional main contribution. The following sections give the
restricted-class bounds and their proofs. Main Section~3 gives the synchronization witnesses for $r=2$ and those for $r=3$ appear in Main Table~\ref{M-tab:small-sync}; Appendix A supplies the
remaining position cases.

\paragraph{Which bounds apply?}
Except for the circular background, the bounds below assume that $a$
has rank $n-1$ and $b$ is an involution. Synchronization is required
unless the cited theorem explicitly says otherwise.
\begin{center}\small
\begin{tabular}{@{}p{9.6cm}l@{}}
\toprule Restriction & Result\\\midrule
One or two transpositions in $b$ & Theorem~\ref{S-thm:sparse}\\
Three transpositions in $b$ & Theorem~\ref{S-thm:three-swaps}\\
Every transposition touches the tail & Theorem~\ref{S-thm:no-rr}\\
At most one transposition has two cyclic endpoints & Theorem~\ref{S-thm:orr-one-cyclic}\\
One $a$-cycle, with the missing image exchanged onto it & Theorem~\ref{S-thm:internal-crossed-hole}\\
Two cyclic transpositions and the specified tail support & Section~\ref{S-sec:two-rr}\\
Aligned tail--cycle support whose index permutation is not an involution & Theorem~\ref{S-thm:aligned-permutations}\\\bottomrule
\end{tabular}
\end{center}
These restrictions overlap without forming one increasing sequence.
The full hypotheses are in the indicated statements. The position-case
proofs in Appendix A are needed for the three-transposition theorem.

\renewcommand{\thesection}{S\arabic{section}}
\renewcommand{\theequation}{S\arabic{equation}}
\renewcommand{\thetable}{S\arabic{table}}
\section{Computations}\label{S-sec:comp}
\begin{StheoremB}
The maximum of $\Tr$ over all synchronizing binary automata with $n$ states equals
$4,5,7,8,9,10,12,13$ for $n=3,\dots,10$, hence $\floor{4n/3}$ for $n=3,4$ and $6\le n\le 10$. Over the
automata with a permutation letter and a rank $n-1$ letter the maximum equals $\floor{4n/3}$ for every
$6\le n\le 12$, namely $8,9,10,12,13,14,16$.
\end{StheoremB}
\subsection{Reduction to a terminal component}
\begin{Slemma}[sink lemma]\label{S-lem:sink}
Let $\mathcal A$ be synchronizing with $n\ge3$ states and let $S$ be its sink strongly connected
component, i.e.\ the unique strongly connected component closed under both letters. Then
$\mathcal A|_S$ is a strongly connected synchronizing automaton, and
\[
\Tr(\mathcal A)\le
\begin{cases}
\Tr(\mathcal A|_S), & |S|\ge3,\\
2, & |S|=2,\\
n, & |S|=1.
\end{cases}
\]
Consequently the maximum of $\Tr$ over synchronizing automata with at most $n$ states is attained
by a strongly connected one (for $n\ge3$, since \v{C}ern\'y's automaton is strongly connected with
$\Tr=n+1$).
\end{Slemma}

\begin{proof}
Behague and Johnson treat the same three cases for their reduction \cite{S-BJ22}; we give the
argument for completeness.
A synchronizing automaton has exactly one strongly connected component closed under the letters
(every state reaches the state in the image of a reset word, and that state reaches a closed
component), and a reset word of $\mathcal A$ is a reset word of $\mathcal A|_S$. If $|S|\ge3$, a
word synchronizing a triple inside $S$ in $\mathcal A|_S$ does the same in $\mathcal A$. If
$|S|=2$, some letter $y$ merges the two states of $S$ (otherwise both letters would permute $S$
and no word could merge them), and some state $u\notin S$ is sent into $S$ by a letter $x$ (every
state reaches $S$); then $x$ maps $\{u\}\cup S$ into $S$ and $xy$ synchronizes this triple. If
$S=\{s\}$, some state $u\ne s$ has $u\cdot x=s$ for a letter $x$, so $x$ has the collision
$\{u,s\}\to s$; if $\Q\cdot x=\{s\}$ then $\Tr=1$, and otherwise Main Lemma~\ref{M-lem:formula} gives
$\Tr\le1+\dep(\{s,w\})$ for any $w\in\Q\cdot x\setminus\{s\}$, where $\dep(\{s,w\})$ is the length
of a shortest word taking $w$ to the absorbing state $s$, at most $n-1$.
\end{proof}

\subsection{The small synchronization cases}
Main Table~\ref{M-tab:small-sync} lists reset words for $G_{3,k}$,
$k=0,1,2$. Their images follow directly from Main equation~\eqref{M-eq:a}.
The reset words at $r=2$ are in Main Section~3. The ancillary program \texttt{boundary\_r2.py} also computes the pair distances for these three automata.

\subsection{Methods}
All computations use the reverse breadth-first search on pairs, which yields the merging distances of all
pairs of an $n$-state automaton in $O(n^2)$ time, and Main Lemma~\ref{M-lem:formula}. Three enumerations were
run on a GPU (CUDA kernels written through CuPy; one automaton per thread).

\emph{All binary automata, $n\le9$.} We enumerate every binary automaton in which all states are
reachable from state $0$, with states numbered in order of discovery (letter $a$ before $b$); every
isomorphism class of such automata occurs at least once, and by Lemma~\ref{S-lem:sink} of this document the maximum of
$\Tr$ over all synchronizing automata with $n$ states is the maximum over these automata with at most
$n$ states. The numbers of enumerated automata are $216$, $5\,248$, $160\,675$, $5\,931\,540$,
$256\,182\,290$, $12\,665\,445\,248$ and $705\,068\,085\,303$ for $n=3,\dots,9$. Write $E(m)$ for the maximum of $\Tr$ over the enumerated synchronizing $m$-state automata
(for $m=10$, over the branch-and-bound enumeration described next), $S(m)$ for the maximum over
strongly connected synchronizing $m$-state automata and $M(m)$ for the maximum over all
synchronizing $m$-state automata. Every strongly connected automaton is reachable from each of its
states, so $S(m)\le E(m)\le M(m)$ (for $m=10$ the first inequality uses that the branch-and-bound
search is exhaustive for $\Tr\ge13$), and Lemma~\ref{S-lem:sink} of this document together with \v{C}ern\'y's
automaton ($S(n)\ge n+1>n$) gives $M(n)\le\max_{3\le m\le n}S(m)\le\max_{3\le m\le n}E(m)$. The
computed values $E(3),\dots,E(10)=4,5,7,8,9,10,12,13$ increase strictly, so
$\max_{m\le n}E(m)=E(n)$ and hence $M(n)=E(n)$ for $3\le n\le10$: the maximum over at most $n$
states is attained with exactly $n$ states, and Table~\ref{S-tab:max} reports $E(n)$. (No
monotonicity of $M$ in $n$ is claimed; adjoining a state can lower $\Tr$.)

\emph{Branch and bound, $n=10$.} The maximum for $n=10$ was obtained by a branch-and-bound
enumeration rooted at a collision image, run with a target value $T$; here $T$ is the value to be
confirmed, $T=13=\floor{40/3}$ for $n=10$. Prefixes of the canonical numbering on $7$ states were
generated on the CPU and pruned when a pair $\{0,w\}$, with $w$ already known to lie in the image of
$a$, had merging distance $<T-1$ in the partial automaton (merging distances in a partial automaton are upper bounds for
merging distances in every completion), and the $2.98\cdot10^{12}$ surviving completions were evaluated on the
GPU. The enumeration is exhaustive for automata with $\Tr\ge T$, because a completion of a pruned
prefix inherits the short pair and satisfies $\Tr\le1+\dep(\{0,w\})<T$ by Main
Lemma~\ref{M-lem:formula}, while nothing with $\Tr<T$ needs to be found. The run found no automaton
with $\Tr\ge14$ and 3 rooted tables with $\Tr=13$, forming 3 isomorphism classes, the same three
classes as in the permutation class; hence $\max\Tr=13=\floor{40/3}$ for $n=10$.
The same enumeration with $T=12$ reproduces the two extremal classes
for $n=9$.

\emph{Automata with a permutation letter and a rank $n-1$ letter.} Here $b$ is taken as one canonical
permutation of each cycle type, the collision image $z$ of $a$ ranges over one state of each cycle
length of $b$ (orbits of the centralizer of $b$), and $a$ ranges over all $(n-1)\binom n2(n-2)!$ maps
with a prescribed collision image and image of size $n-1$. This class contains all extremal automata
for $n\le9$ and was enumerated exhaustively up to $n=12$ ($5.1\cdot10^{11}$ automata),
and for $n=13$ over $18$ selected cycle types of $b$ (all types with two cycles, the types with three
cycles and smallest cycle of length at least $2$, and $9{+}3{+}1$, $8{+}4{+}1$, $6{+}6{+}1$,
$5{+}5{+}2{+}1$); for $n=13$ the maximum over the selected types is $17=\floor{52/3}$
($1.6\cdot10^{12}$ automata), attained by $3$ classes, one of them $\G{4}{1}$.

\emph{Finite validation.} Direct evaluation of $\lab$ from Main equation~\eqref{M-eq:lambda}
checks the edge inequalities of Main Lemma~\ref{M-lem:lip} and the
hypotheses of Main Lemma~\ref{M-lem:potential} on each finite automaton. Breadth-first search over all $3$-subsets confirms $\Tr(\G{r}{k})=4r+k$, for $3\le r\le15$ and $0\le k\le2$ ($n\le47$). A separate check follows the
pair-merging strategy of the proof of Main Lemma~\ref{M-lem:sync} move by move from every pair, for
$4\le r\le 12$ and $k\le 2$, and the word of Main Theorem~\ref{M-thm:upper} was applied to the stated
triple for $3\le r\le11$. The ancillary program treats the $r=2$ triples.

\subsection{Results}
Table~\ref{S-tab:max} lists the maxima. For each $n$ the second column is the maximum over all
synchronizing binary automata with $n$ states, the third the number of isomorphism classes (up to
renaming the letters) attaining it, the fourth the maximum over the permutation/rank-$(n-1)$ class
and the fifth the number of classes attaining it there.

\begin{table}[ht]
\centering
\begin{tabular}{@{}rrrrrr@{}}
\toprule
$n$ & $\max\Tr$ & classes & $\max\Tr$, perm.\ class & classes & $\floor{4n/3}$\\
\midrule
3 & 4 & 2 & 4 & 2 & 4\\
4 & 5 & 4 & 5 & 4 & 5\\
5 & 7 & 1 & 7 & 1 & 6\\
6 & 8 & 2 & 8 & 2 & 8\\
7 & 9 & 12 & 9 & 12 & 9\\
8 & 10 & 24 & 10 & 24 & 10\\
9 & 12 & 2 & 12 & 2 & 12\\
10 & 13 & 3 & 13 & 3 & 13\\
11 & & & 14 & 9 & 14\\
12 & & & 16 & 2 & 16\\
13 & & & 17 & 3 & 17\\
\bottomrule
\end{tabular}
\medskip
\caption{The maximum triple rendezvous time of binary $n$-state synchronizing automata (all
automata for $n\le 10$; automata with a permutation letter and a rank $n-1$ letter for $n\le 13$,
restricted cycle types for $n=13$).}
\label{S-tab:max}
\end{table}

All extremal automata for $3\le n\le10$ have a permutation letter and a letter of rank $n-1$, and the
values agree with $\floor{4n/3}$ except for $n=5$, where the maximum $7$ exceeds $\floor{20/3}=6$. For
$n=9$ the two extremal classes are Gonze and Jungers' $TR_9$ and $\G{3}{0}$; for $n=10$, $11$, $12$
the members $\G{3}{1}$, $\G{3}{2}$ (both variants in which the two extra states are fixed by $b$ or
swapped by $b$) and $\G{4}{0}$ are among the extremal classes. The second extremal class for $n=12$
is $a=(0\,1)(2\,3\,10\,11)(4\,9)(7\,8)$, $5\to6\to0$, $b=(0\,1\,2\,3\,4)(5\,6\,7\,8\,9)(10\,11)$, again
with $\Tr=16$; we do not know whether it extends to a family.

\subsection{The permutation class}
Among strongly connected synchronizing binary automata with a
permutation letter and a rank-$(n-1)$ letter, the maximum of
$\tau_A(q)$ over all states and automata is $n+1$ for $n=7$ and
$n+2$ for $8\le n\le11$. Table~\ref{S-tab:max} gives maximum
$T_3$ for the synchronizing permutation/rank-$(n-1)$ class:
$10,12,13,14$ at $n=8,9,10,11$. Thus the two reported maxima
agree at $n=8$ and differ by one at $n=9,10,11$.

For all strongly connected synchronizing binary automata, the finite
maximum of the specified-state threshold is $10,12,13$ at $n=8,9,10$
\cite[Theorem E and Supplement J.1]{S-Path}. That enumeration
roots a transition table at the state whose threshold is tested.
At $n=10$, a partial table is discarded only when a completed
path in its pair graph already merges this state with a partner
in fewer than $13$ letters; each completion inherits that word.
The complete ranges and counts for this search are given with
Theorem~E. Its extrema may be attained by automata outside the
permutation class.

\subsection{Larger subsets}\label{S-sec:largerk}
Let $T_k$ denote the shortest length of a word synchronizing some $k$-subset, so that
$\operatorname{rdv}(k,n)$ of \cite{S-BJ22} is the maximum of $T_k$ over all synchronizing $n$-state
automata, with no restriction on the alphabet. Our enumeration covers the class of strongly
connected synchronizing \emph{binary} automata, and we write $\operatorname{rdv}_{2,\mathrm{sc}}(k,n)$
for the maximum of $T_k$ over that class; the reduction of Lemma~\ref{S-lem:sink} of this document to strongly
connected automata is not applied here, since for $T_k$ it would require a separate treatment of
sinks with fewer than $k$ states. We computed $T_k$ for all $4\le k\le n\le 8$ over all strongly
connected synchronizing binary automata by a breadth-first search over the subsets of size at most
$k$ (one automaton per GPU thread). In every case the maximum is
\[
  \operatorname{rdv}_{2,\mathrm{sc}}(k,n)=(k-2)n+1\qquad(4\le k\le n\le 8),
\]
that is, exactly \v{C}ern\'y's lower bound of \cite{S-BJ22}; for $k=n$ it is the maximal reset threshold
$(n-1)^2$ of strongly connected binary automata with $n\le8$ states. For $6\le n\le 8$ the
\v{C}ern\'y automaton is the only automaton attaining the maximum, except for $(n,k)=(6,6)$, where a
second automaton with reset threshold $25$ is known; for $n=5$, $k=4$ there are four extremal
classes. For $n=9$ and $k=4$ the enumeration of all $7.05\cdot10^{11}$ canonical tables gives the
maximum $19=2n+1$, attained by one class. The histograms are gapped: for $n=8$ and $k=4$ no
automaton has $T_4=16$. In \v{C}ern\'y's automaton $C_n$ ($a$ a cyclic permutation, $b$ fixing all
states but one) we verified $T_k(C_n)=(k-2)n+1$ for all $3\le k\le n\le 12$. Thus the triple case is exceptional: \v{C}ern\'y's automaton gives $n+1$,
whereas Main Theorem~A supplies the lower construction $\floor{4n/3}$ for
every $n\ge6$. A matching general upper bound remains open.

\begin{Squestion}\label{S-conj:k}
For all $4\le k\le n$, is
$\operatorname{rdv}_{2,\mathrm{sc}}(k,n)=(k-2)n+1$?
\end{Squestion}
At $k=n$ this is the strongly connected binary instance of the
\v{C}ern\'y conjecture. For $k=4,5$ the general upper estimates
of \cite{S-BJ22} are quadratic.

\section{An upper bound for circular automata}\label{S-sec:circ}

For circular synchronizing automata the bound $\Tr\le n+1$ already follows from
Dubuc~\cite[Proposition~4.6]{S-Dubuc98}: every nonempty proper subset can be extended by an inverse
word of length at most $n$. Choose a collision fibre $S=x^{-1}(\{z\})$ of a non-permutation
letter $x$. If $|S|\ge3$, then $\Tr=1$. Otherwise $|S|=2$; an extending word $u$ has
$|S\cdot u^{-1}|\ge3$ and $|u|\le n$, so $ux$ merges at least three states. Thus no restriction
on the ranks of the other letters is needed. We give an elementary constructive proof in the
rank-$(n-1)$ case, which exposes the rotation obstruction used below. Throughout,
$b$ is a permutation of $\Q$ and $a$ has rank $n-1$, so $a$ has exactly one collision pair
$\{p,q\}$, $a(p)=a(q)=z$, and exactly one state $h\notin a(\Q)$, the \emph{hole}.

\begin{Slemma}\label{S-lem:dist}
For every pair $P$ of states, $\dep(P)=1+\operatorname{dist}(P,\{p,q\})$, where $\operatorname{dist}$
denotes the length of a shortest word mapping $P$ onto $\{p,q\}$. Consequently
\begin{equation}\label{S-eq:t3perm}
\Tr=2+\min\{\operatorname{dist}(\{z,y\},\{p,q\}) : y\in\Q\setminus\{z,h\}\}.
\end{equation}
\end{Slemma}

\begin{proof}
The letter $b$ merges no pair and $a$ merges exactly the pair $\{p,q\}$, so a word merging $P$ ends
with the letter $a$ applied to $\{p,q\}$. The formula follows from \MainPrefix Lemma~\ref{M-lem:formula}, since
$a(\Q)\setminus\{z\}=\Q\setminus\{z,h\}$.
\end{proof}

It is convenient to read \eqref{S-eq:t3perm} dynamically. Put a token on every state and apply a word
$u$; a state carries $|u^{-1}(s)|$ tokens, so it is a \emph{hole}, a \emph{single} or a \emph{double}.
After the first letter $a$ the only double sits on $z$ and the only hole on $h$. The letter $b$ moves
all tokens bijectively. The letter $a$ moves the tokens of $p$ and $q$ together to $z$; a triple
appears exactly when one of $p,q$ carries a double and the other is not a hole. Thus
$\Tr\le 2+|v|$ for every word $v$ with $z\cdot v=p$ and $q\in\Q\cdot(av)$ (or with $p$ and $q$
exchanged). For $s,t$ in one $b$-cycle let $\operatorname{rot}(s\to t)$ be the number of $b$-steps
from $s$ to $t$, and let $\ell_s$ be the length of the $b$-cycle of $s$.

\begin{Slemma}[rotation lemma]\label{S-lem:rot}
Suppose $z$ lies in the $b$-cycle of $p$ and put $j_1=\operatorname{rot}(z\to p)$.
If $h\cdot b^{j_1}\ne q$ then $\Tr\le j_1+2\le\ell_p+1$; if $h\cdot b^{j_1}=q$ but $\ell_q\nmid\ell_p$
then $\Tr\le j_1+\ell_p+2\le 2\ell_p+1$. The same holds with $p$ and $q$ exchanged.
\end{Slemma}

\begin{proof}
The word $b^{j_1}$ maps $\{z,y\}$ with $y=q\cdot b^{-j_1}$ onto $\{p,q\}$; here $y\ne z$ since
$q\ne p$, and $y\ne h$ by hypothesis, so \eqref{S-eq:t3perm} gives the first bound. In the second case
$b^{j_1+\ell_p}$ maps $\{z,y'\}$, $y'=q\cdot b^{-j_1-\ell_p}=h\cdot b^{-\ell_p}$, onto $\{p,q\}$, and
$y'\ne h$ because $\ell_q\nmid\ell_p$.
\end{proof}

The remaining situation, $h\cdot b^{j_1}=q$ with $\ell_q\mid\ell_p$, is a \emph{lock}: the double
started at $z$ and the hole started at $h$ move in step under $b$, and the hole stands on $q$
whenever the double stands on $p$. The family $\G{r}{0}$ exhibits this lock ($\ell_q=r$ divides
$\ell_p=2r$, and the hole $0'$ is aligned with $z$). For $k>0$ the initial collision image lies
on the added chain, so this initial pure-rotation argument does not apply to it.
The exact value $4r+k$ is established by the lower and upper bound proofs
in the main article.

\begin{Stheorem}\label{S-thm:circ}
Let $A$ be a synchronizing automaton with $n\ge3$ states in which $b$ is a cyclic permutation of
$\Q$ and $a$ has rank $n-1$. Then $\Tr(A)\le n+1$, and \v{C}ern\'y's automaton attains the bound.
\end{Stheorem}

\begin{proof}
Identify $\Q$ with $\mathbb Z_n$ so that $b(i)=i+1$. By Lemma~\ref{S-lem:rot}, applied to $p$ and to
$q$, we have $\Tr\le n+1$ unless $h=z+(q-p)$ and $h=z+(p-q)$. In the latter case $2(q-p)=0$, so $n$
is even, $q=p+n/2$ and $h=z+n/2$. Write $\sigma(s)=s+n/2$. Call an antipodal pair $\{s,\sigma(s)\}$
other than $\{p,q\}$ \emph{preserved} if $a(\sigma(s))=\sigma(a(s))$.

Not every such pair is preserved. Otherwise, as $a$ restricted to $\Q\setminus\{p,q\}$ is a
bijection onto $\Q\setminus\{z,h\}$ and both sets are unions of $n/2-1$ antipodal pairs, $a$ would
induce a bijection between these antipodal pairs and hence would map every non-antipodal pair of
$\Q\setminus\{p,q\}$ to a non-antipodal pair; a pair $\{p,t\}$ or $\{q,t\}$ with $t\notin\{p,q\}$ is
mapped to $\{z,a(t)\}$, which is not antipodal because $a(t)\ne h$; and $b$ preserves distances.
So no non-antipodal pair would ever become antipodal, and since every merge passes through the
antipodal pair $\{p,q\}$, no non-antipodal pair would ever merge, contradicting synchronization
($n\ge4$ here).

Fix a non-preserved pair $\{s,\sigma(s)\}$ and name its elements so that $i:=s-z\in[0,n/2-1]$.
Put $u=a(s)$ and let $j\in[0,n/2-1]$ be the smaller of $p-u$ and $q-u$ modulo $n$. We claim that
$w=a\,b^{i}\,a\,b^{j}\,a$ merges three states; as $|w|=i+j+3\le n+1$ this proves the theorem.
After $a\,b^i$ the double is on $s$ and the only hole on $\sigma(s)$; neither lies in $\{p,q\}$.
After the second $a$ this double is on $u$ (a second double appears on $z$, since $p$ and $q$ were
singles) and the holes are $h$ and $a(\sigma(s))$. After $b^j$ the first double is on $p$ or on $q$,
and the partner of that state is $\sigma(u)+j$, which is a hole only
if $\sigma(u)=h=\sigma(z)$, i.e.\ $a(s)=z$, impossible as $s\notin\{p,q\}$, or if
$\sigma(u)=a(\sigma(s))$, which is excluded since the pair is not preserved. So the partner is
occupied and the last $a$ merges the double with the partner's token.

\v{C}ern\'y's automaton attains the bound: Behague and Johnson
\cite[Section~1, p.~4]{S-BJ22} give the minimum merging length over all $k$-subsets of $C_n$ as
$(k-2)n+1$. Taking $k=3$ therefore yields $\Tr(C_n)=n+1$.
\end{proof}

The proof constructs a triple-merging word of length at most $n+1$, of the form $a b^j a$
in the unlocked case or $a b^i a b^j a$ in the locked case. It does not claim that a shortest
word must have either form. A script replays the construction on all $1.2\cdot10^6$ synchronizing circular
automata with a rank-$(n-1)$ letter for $n\le8$. The rank
restriction is used only through the uniqueness of the collision pair and of the hole.
Exhaustively, the maximum of $\Tr$ over all circular synchronizing automata (any second letter,
$n^n$ maps) equals $n+1$ for $5\le n\le11$, and for $5\le n\le 11$ it is
attained only when the second letter has rank $n-1$; over circular automata with a rank-$(n-1)$
letter it equals $n+1$ for $7\le n\le12$ (the circular type; Table~\ref{S-tab:types} lists selected cycle-type maxima for $n=9,10,11$). The
arbitrary-rank upper bound is the consequence of Dubuc's theorem established at the start of this
section.

\begin{table}[ht]
\centering
\begin{tabular}{@{}lrrr@{}}
\toprule
cycle type of $b$ & $n=9$ & $n=10$ & $n=11$\\
\midrule
$n$ (circular) & 10 & 11 & 12\\
$(n-1)+1$ & 11 & 11 & 13\\
$(n-2)+2$ & 10 & 12 & 12\\
$(n-3)+3$ & 12 & 11 & 13\\
$(n-4)+4$ & 10 & 11 & 13\\
$(n-5)+5$ & -- & 11 & 12\\
maximum over all types & 12 & 13 & 14\\
\bottomrule
\end{tabular}
\medskip
\caption{Maximum of $\Tr$ over automata with a permutation letter $b$ of the given cycle type and
a letter $a$ of rank $n-1$ (exhaustive). The identity type gives
$\Tr=2$. The maximum is attained by the types $6{+}3$ and $3{+}3{+}2{+}1$ ($n=9$), $6{+}3{+}1$ and
$3{+}3{+}2{+}1{+}1$ ($n=10$), and by seven types for $n=11$, among them $8{+}2{+}1$, $6{+}3{+}2$ and
the involution type $2{+}2{+}2{+}2{+}2{+}1$.}
\label{S-tab:types}
\end{table}

For two $b$-cycles of lengths $\ell_1\ge\ell_2$ the data for $9\le n\le 12$ (Table~\ref{S-tab:types};
for $n=12$ the exhaustive enumeration of the permutation class in Section~\ref{S-sec:comp}) satisfy
$\Tr\le\ell_1+2\ell_2+1$, with equality for the types $8{+}1$, $10{+}1$ and $10{+}2$; the members
$\G{r}{0}$ have $\ell_1=2\ell_2$ and $\Tr=4r=\ell_1+2\ell_2$. Types with many short cycles behave
differently: for $b$ an involution the maximum is $n+3$ ($n=11,12$). These examples rule out a
uniform $n+1$ bound for the whole permutation class and show that short $b$-cycles need not
make $\Tr$ small; the actions of $a$ also matter.

\section{Permutation letters with one or two transpositions}\label{S-sec:sparse}

We next restrict the support of the permutation letter rather than requiring
it to be one cycle. Throughout this section, $a$ has rank $n-1$ and $b$ is
an involution. The arguments use full inverse images of pairs, as in the
standard subset approach to rendezvous times~\cite{S-GJ16,S-BJ22}; neither the
inverse-image method nor the near-permutation setting is introduced here.
Near-permutation semigroups have also been studied from the different
viewpoint of inverse-semigroup structure~\cite{S-AndreInverse04}. We retain
the original binary word length throughout. The subcase in which $z$ and
$h$ lie in the same transposition orbit of $b$ is already covered by
the extension tools of Zhu~\cite[Propositions~9--10]{S-Zhu24}, which give
$\Tr\le n+1$ there. The argument below also treats the cross-orbit cases.
Under additional subgroup hypotheses, Andr\'e describes idempotents with
triple fibres and reductions to local monoids~\cite[Lemma~5.13 and
Proposition~5.14]{S-AndreRegular04}. Those are structural predecessors;
the bounds here count steps in the original binary alphabet.

\begin{Stheorem}\label{S-thm:sparse}
Let $\mathcal A=(Q,\{a,b\})$ be synchronizing and suppose
$\operatorname{rank}(a)=n-1$.
If $b$ is a single transposition, then $\Tr(\mathcal A)\le n$ for $n\ge3$.
If $b$ is a product of exactly two disjoint transpositions, then
$\Tr(\mathcal A)\le n+1$ for $n\ge4$.
\end{Stheorem}

The bounds concern the existence of a word merging \emph{some} three
states. In particular, no linear bound on the reset threshold is asserted.
The proof is uniform in all cycle and tail lengths; the finite experiments
reported with the proof files are not premises of the theorem.

\subsection{Inverse pairs and the functional graph}
Write $h$ for the missing image of $a$, $z$ for its double image, and
$K=a^{-1}(\{z\})$ for its collision pair. For every pair $P$,
\begin{equation}\label{S-eq:sparse-size}
 |a^{-1}(P)|=2+\boldsymbol1_{z\in P}-\boldsymbol1_{h\in P}.
\end{equation}
A pair containing $z$ and avoiding $h$ is therefore an expansion target.
In a shortest inverse trajectory from a singleton to a set of size at
least three, the first growth to size two occurs at $\{z\}$ under $a^{-1}$.
Any earlier singleton steps can be removed. A later contraction to a
singleton likewise allows a shorter restart, and the empty set cannot
recover. Hence
\begin{equation}\label{S-eq:sparse-pairdistance}
 \Tr=2+\operatorname{dist}(K,\{P:z\in P,\ h\notin P\}),
\end{equation}
where the distance uses only size-preserving inverse-pair steps under
$a$ and $b$. An unreachable target means that no triple can merge.
We use $A(P)=a^{-1}(P)$ and $B(P)=b(P)$ below; sequences of these operators
are described in the order applied. Reversing the corresponding letters
gives an ordinary forward word of the same length.

The indegrees of the functional graph of $a$ are one zero, one two, and
otherwise one. Thus there is a unique noncyclic tail
\[
 T_0=h\longrightarrow T_1\longrightarrow\cdots\longrightarrow
 T_{t-1}\longrightarrow C_0=z
\]
entering a cycle $C=(C_0,\ldots,C_{c-1})$ of length $c$; all other
components, if present, are cycles. Here $t,c\ge1$ and
$K=\{T_{t-1},C_{c-1}\}$. If $t>c$, the word $a^{c+1}$ maps the three
distinct states $C_{c-1},T_{t-1},T_{t-c-1}$ to $z$, so $\Tr\le c+1$.
When $t\le c$, put $v=c-t$. The initial inverse-pair chain is
\begin{equation}\label{S-eq:sparse-chain}
 P_i=A^{t-i}(\{z\})=\{T_i,C_{v+i}\},\qquad 0\le i<t.
\end{equation}
For $i>0$, $A(P_i)=P_{i-1}$. At $P_0$ it contracts to a singleton if
$v>0$, and returns to $K$ if $v=0$. No $P_i$ is an expansion target.
Once a pair of cyclic states touching $C$ is reached, at most $c-1$
further $A$ steps reach an expansion target: stop at the first visit
to $z$ and then apply one more $A$.

\subsection{One transposition}
Let $R$ be the set of cyclic states of $a$. If $b$ preserves $R$, both
letters permute $R$, so synchronization forces $|R|=1$. Then $a^2$
merges three states. Otherwise $b$ exchanges a tail state $T_x$ with a
cyclic state $Y$. If $Y\notin C$, both letters permute the invariant set
$C$, which must be a singleton. For $t\ge2$, $a^2$ again suffices; for
$t=1$, $b$ fixes $z$ and the word $aba$ merges three states.

If $Y\in C$, synchronization excludes additional disjoint cycles, so
$n=t+c$. The long-tail case was handled above. Suppose $t\le c$, and
write $Y=C_y$. If $y=v+x$, the transposition fixes each pair in
\eqref{S-eq:sparse-chain}, including the pair whose endpoints it exchanges.
The chain is closed apart from a singleton contraction or a return to
$K$, so \eqref{S-eq:sparse-pairdistance} excludes synchronization.
Otherwise $B(P_x)=\{C_y,C_{v+x}\}$ consists of distinct cyclic states.
It reaches a target in $j=\min(y,v+x)$ inverse $a$ steps. Including the
initial $t-x$ steps, the $B$ step and the final expansion, the length is
\[
 t-x+j+2\le c+2.
\]
This is at most $n$ when $t\ge2$. If $t=1$, then $x=0$ and
$y\ne v=c-1$, whence $j=y\le c-2$ and the length is at most $c+1=n$.
This proves the first assertion of Theorem~\ref{S-thm:sparse}, without a
strong-connectivity assumption.

\subsection{Two transpositions with a one-state tail}
We first cover $t=1$ without assuming strong connectivity. The restriction
of $a$ to $R=Q\setminus\{h\}$ is a permutation $\pi$. Put
$p=C_{c-1}=\pi^{-1}(z)$, so $K=\{h,p\}$.
Synchronization excludes $b(h)=h$, since then both letters permute $R$.
It also excludes $b(h)=p$: in that case $B$ fixes $K$, whereas $A(K)$
is either a singleton or $K$ itself. Consequently $B(K)$ is a cyclic
pair. If it touches $C$, a triple is reached in at most
$1+1+(c-1)+1=c+2\le n+1$ steps.

The remaining configuration is necessarily
\[
 b=(h,X)(p,Y),\qquad X,Y\notin C.
\]
The points $X,Y$ must belong to the same $\pi$-cycle $D$: otherwise
$C$ together with the cycle of $Y$ is invariant and permuted by both
letters, contradicting synchronization. Further cycles are disjoint
closed components. Hence $n=1+c+d$, where $d=|D|$.
Let $X\pi^\Delta=Y$, with $1\le\Delta<d$. If $2\Delta\not\equiv0
\pmod d$, set $j=d-\Delta$. Starting from $K$, apply $B,A^j,B$ to get
successively
\[
 \{X,Y\},\quad\{Y,W\},\quad\{p,W\},\qquad W\in D\setminus\{X,Y\}.
\]
After $c-1$ further $A$ steps and the final expansion a triple is reached.
The total length, including the initial $A$ from $\{z\}$, is
$c+j+3\le c+d+2=n+1$.
If $2\Delta\equiv0\pmod d$, the points are antipodal. The antipodal
pairs of $D$ are closed under $A$; $B$ fixes them except that
$\{X,Y\}$ returns to $K$. Together with $K$ this is a closed inverse-pair
graph without a target. Equation~\eqref{S-eq:sparse-pairdistance} again
contradicts synchronization. This completes the tail-one case.

\subsection{Reduction to strong connectivity}
We may now assume $2\le t\le c$. A synchronizing automaton has a unique
terminal strongly connected component $U$, and its restriction to $U$
is synchronizing. Because $b$ is a permutation and $U$ is closed, $b$
permutes $U$. The component cannot be a singleton: such a point would
be fixed by both letters, lie outside the nontrivial receiving cycle,
and have itself as its only preimage under either letter, making it
unreachable from outside.
If $a|_U$ were injective, both letters would permute $U$, contradicting
synchronization. Thus $K\subset U$, so $C\subset U$ and $|U|\ge c+1\ge3$;
also $\operatorname{rank}(a|_U)=|U|-1$.
If $b|_U$ is a single transposition, the bound just proved applies.
If it is the identity, the cycle $C$ cannot synchronize. Otherwise it
has two transpositions, and it suffices to prove the bound on $U$.
The restricted tail may be shorter; the long-tail and tail-one arguments
apply again if necessary. We can therefore assume strong connectivity
and $2\le t\le c$ in what follows. Strong connectivity forces $h$ to
be moved by $b$, since $h$ has no incoming $a$ edge.

Let $k$ be the number of moved states on the tail. Then $k\ge1$.
The case $k=4$ preserves $R$, and both letters would permute this set
of size at least two. For the same reason, when $k=2$ the two tail states
cannot be exchanged with each other. Thus $k=1,2,3$, with precisely the
following possibilities.

\subsection{One moved tail state}
Here the only moved tail state is $h$.
First suppose $b(h)=C_x$. If $x\ne v$, then $B(P_0)$ is a cyclic pair
touching $C$, yielding length at most $t+c+1\le n+1$.
Suppose $x=v$. If the second transposition avoids
$C_{v+1},\ldots,C_{c-1}$, it fixes the entire initial chain, so no target
is reachable. Otherwise it exchanges $C_{v+w}$, for some $1\le w<t$,
with a different cyclic state $V$. From $\{z\}$ take $A^{t-w},B$,
reaching $\{T_w,V\}$. If the cyclic coordinate reaches $z$ in fewer than
$w$ further $A$ steps, the next step expands before the tail reaches $h$,
at total cost at most $t+1$.
Otherwise $w$ steps reach $\{h,V'\}$ with $V'\ne C_v$.
Indeed this is automatic on another cycle; on $C$, write $V=C_q$.
Absence of an earlier visit to $z$ gives $q\ge w$, and
$q-w=v$ would make $V=C_{v+w}$, contrary to the transposition having
distinct endpoints. A further $B$ now gives a cyclic pair containing
$C_v$. At most $v$ further $A$ steps and one expansion suffice, for
total cost at most $t+v+3=c+3\le n+1$.

Next suppose $X=b(h)$ lies on another cycle. The second transposition
must cross from $C$ to an external cycle, since otherwise both letters
permute $C$. Its external endpoint $Y$ must be on the same cycle $D$
as $X$: if not, $C$ and the cycle of $Y$ form an invariant set permuted
by both letters. No further cycle can occur. Thus
\[
 b=(h,X)(C_y,Y),\qquad n=t+c+d,
\]
where $X,Y\in D$ and $d=|D|$.
If $y\ne v$, the pair $B(P_0)=\{X,C_v\}$ expands after $v+1$ more
$A$ steps, for length $c+2$. If $y=v$, it becomes $\{X,Y\}$ on $D$.
The same non-antipodal rotation as in the tail-one argument takes this
pair to $\{Y,W\}$ in some $j\le d-1$ steps, and $B$ then gives
$\{C_v,W\}$. The total length through the final expansion is
$t+1+j+1+v+1=c+j+3\le c+d+2\le n+1$.
In the antipodal case, the antipodal pairs on $D$ and the initial chain
$\{P_0,\ldots,P_{t-1}\}$ form a closed inverse-pair graph: only
$\{X,Y\}$ and $P_0$ are exchanged by $B$, and all other pairs are
fixed. There is no target, contrary to synchronization.

\subsection{Two moved tail states}
The transpositions have the form $(h,X)(T_i,Y)$, where $0<i<t$ and
$X,Y$ are distinct cyclic states. If $C_v\notin\{X,Y\}$, then
$B(P_0)=\{X,C_v\}$ is a cyclic pair touching $C$. If
$C_{v+i}\notin\{X,Y\}$, the same holds for
$B(P_i)=\{Y,C_{v+i}\}$. These give lengths at most $t+c+1$ and
$t+c+1-i$, respectively, both at most $n+1$.
If neither condition holds, distinctness forces
$\{X,Y\}=\{C_v,C_{v+i}\}$. According to the assignment of the endpoints,
$B$ either fixes both $P_0,P_i$ or exchanges them as whole pairs; it
fixes every other $P_j$. Thus the initial chain is closed without a
target, which synchronization excludes. This covers both alignments,
including the possibility of exchanging two different initial pairs.

\subsection{Three moved tail states}
There is one tail--tail exchange and one tail--cycle exchange
$(T_i,C_x)$. Strong connectivity forces the cyclic endpoint onto $C$
and excludes additional cycles. Hence $n=t+c$ and $t\ge3$.
If $x\ne v+i$, then $A^{t-i},B$ produces the cyclic pair
$\{C_x,C_{v+i}\}$. With $\min(x,v+i)$ more $A$ steps and an expansion,
the length is at most $c+2\le n$. Only $x=v+i$ remains.

If the cross edge is $(h,C_v)$, write the tail exchange as
$(T_i,T_j)$ with $1\le i<j<t$. The sequence
$A^{t-j},B,A^i,B$ gives
\[
 \{T_j,C_{v+j}\},\quad\{T_i,C_{v+j}\},\quad
 \{h,C_{v+j-i}\},\quad\{C_v,C_{v+j-i}\}.
\]
All displayed cycle indices lie between $0$ and $c-1$. Apply $A^v$
and one final $A$. The length is $c+3-(j-i)\le c+2\le n$.

Otherwise the tail exchange is $(h,T_j)$ and the cross edge is
$(T_i,C_{v+i})$. Strong connectivity forces $0<i<j<t$: a path from
$C$ to $h$ must first enter at $T_i$ and reach $T_j$ before using its
edge to $h$. Put $d=j-i$. Starting with $A^t,B$ gives
$\{T_j,C_v\}$. If $v\le d$, use $A^v$ to reach
$\{T_{j-v},z\}$. Its tail coordinate is at least $i\ge1$, so a final
$A$ expands to three, at cost $c+2\le n$.
If $v>d$, use $A^d,B$ to obtain $\{C_{v+i},C_{v-d}\}$.
The two indices differ by $j<c$, and $v-d>0$, so these states are
distinct. Apply $A^{v-d}$ and one final $A$. The total length is
$c+3\le t+c=n$.
All possibilities have now been covered, proving Theorem~\ref{S-thm:sparse}.

\subsection{Sharpness of the two-transposition bound}
\begin{Sproposition}[Sharpness for two transpositions]
\label{S-prop:two-transpositions-sharp}
For every $n\ge5$, there is a synchronizing binary automaton on $n$ states
whose letter $a$ has rank $n-1$, whose letter $b$ is a product of exactly
two disjoint transpositions, and whose triple rendezvous time is $n+1$.
\end{Sproposition}

\begin{proof}
Put $m=n-2\ge3$ and take
$Q=\{0,1,\ldots,m-1,x,y\}$. Let $a$ cyclically permute
$0,1,\ldots,m-1$ and satisfy $xa=ya=y$, and let
$b=(m-1\;x)(0\;y)$, fixing the other states.
Then $a$ has rank $m+1=n-1$ and $b$ has the required form.
The word $ab a^{m-1}ba$ maps $0,x,y$ to $y$, so $T_3\le m+3$.

For the reverse inequality, trace full inverse images from a singleton
until the first set of size three. The only first pair is
$K=\{x,y\}$, reached from $\{y\}$ by $a^{-1}$.
Returns to a singleton can be discarded by restarting there, and pair
loops can be deleted. At $K$, inverse $a$ is a loop and inverse $b$
gives $E_0=\{m-1,0\}$. The ordinary adjacent pairs
\[
 E_j=\{m-1-j\pmod m,-j\pmod m\},\qquad 0\le j<m,
\]
form a directed cycle under inverse $a$.
The letter $b$ returns $E_0$ to $K$, sends
$E_1=\{m-2,m-1\}$ to $\{m-2,x\}$, and sends
$E_{m-1}=\{0,1\}$ to $\{y,1\}$; it fixes every other $E_j$.
At $\{m-2,x\}$, inverse $a$ contracts to a singleton and $b$ returns
to $E_1$. At $\{y,1\}$, inverse $a$ expands to $\{x,y,0\}$ and $b$
returns to $E_{m-1}$. This lists all reachable pair types, also for $m=3$.
Consequently the first expansion needs the initial $a,b$, at least $m-1$
steps of inverse $a$ from $E_0$ to $E_{m-1}$, and the final $b,a$.
Thus $T_3\ge m+3$, proving equality.

It remains to verify synchronization. Set $c=ab$. Its unique cycle is
\[
 C'=(0,1,\ldots,m-2,x),
\]
and its tail is $m-1\to y\to0$. The map $c^m$ fixes $C'$ pointwise
and maps $Q$ onto $C'$. On this cycle, $f=ac^m$ has the action
\[
 i f=i+1\quad(0\le i\le m-3),\qquad (m-2)f=m-2,\quad xf=x.
\]
Thus $Qc^m f^{m-2}=\{m-2,x\}$. Applying $c^{m-1}$ gives
$\{m-3,m-2\}$, which $f$ merges. An explicit reset word is therefore
\[ c^m f^{m-2}c^{m-1}f. \qedhere \]
\end{proof}

\noindent
This construction adjoins two states to the classical \v{C}ern\'y automaton.
On the $m$ original states, that automaton is generated by the cyclic
permutation and the identity modified at $m-1$ to map it to $0$.
Its triple time is $m+1$~\cite[Section~3, footnote~6]{S-GJ16}.

\section{Further bounds for involution letters}\label{S-sec:three-swaps}

The number of transpositions can be relaxed substantially when their
positions in the functional graph of $a$ are restricted. We first give
such a bound, and then complete the case of three arbitrary disjoint
transpositions. We continue to count original binary letters.

\begin{Slemma}\label{S-lem:three-terminal}
Let $(Q,\{a,b\})$ be synchronizing, with $n=|Q|\ge3$,
$\operatorname{rank}(a)=n-1$, and $b$ an involution. Let $U$ be its unique
terminal strongly connected component. If $|U|\le2$, then $\Tr\le3$.
If $|U|\ge3$, the restriction is synchronizing,
$\operatorname{rank}(a|_U)=|U|-1$, and any triple word for the restriction
is a triple word for $Q$. The transpositions of $b|_U$ are precisely the
transpositions whose endpoints lie in $U$. Cyclicity under $a$ is unchanged
by this restriction.
\end{Slemma}
\begin{proof}
Both letters preserve $U$. A permutation preserving a finite set cannot
map a point outside it into it, so no $b$-edge crosses its boundary.
For $|U|\ge2$, injectivity of $a|_U$ would make both letters permutations
on $U$, contradicting synchronization. Thus the global collision pair is
in $U$, and the restriction has rank $|U|-1$.

If $|U|=2$, $a|_U$ is constant at a point $z$. Some $a$-edge from
outside enters $U$, and it must enter the other point $v$, since $z$
already has its two preimages in $U$. Hence $a^2$ merges that outside
point, $v$ and $z$.
If $|U|=1$, its sink is the double image $z$: otherwise its only
$a$-preimage would be itself, and no outside state could reach it.
Its receiving cycle has length one. A tail of length at least two gives
$a^2$. With tail $h\to z$, the word $aba$ suffices if $b(h)\ne h$:
its inverse sequence is $\{z\},\{h,z\},\{b(h),z\}$ followed by an
expansion. If $b(h)=h$, the set $\{h,z\}$ has no incoming transition
from the other states, contradicting synchronization.
The remaining assertions follow from restriction to a closed subset;
in particular an $a$-orbit is cyclic in the restriction exactly when
it was cyclic in $Q$.
\end{proof}

\begin{Stheorem}\label{S-thm:no-rr}
Let $\mathcal A=(Q,\{a,b\})$ be synchronizing on $n\ge3$ states, with
$\operatorname{rank}(a)=n-1$ and $b$ an involution.
If every transposition of $b$ has at least one endpoint on the noncyclic
tail of $a$, then $\Tr(\mathcal A)\le n$.
\end{Stheorem}

\begin{proof}
Write the tail and its receiving cycle as
\[
 T_0=h\to\cdots\to T_{t-1}\to C_0=z,
 \qquad C=(C_0,\ldots,C_{c-1}).
\]
If $b$ is the identity, synchronization forces $c=1$, $t=n-1$, and no
other cycle, so $a^2$ merges three states. A single transposition is
covered by Theorem~\ref{S-thm:sparse}. Also, if $t>c$, the word $a^{c+1}$
merges $C_{c-1},T_{t-1},T_{t-c-1}$. We may assume at least two
transpositions and $t\le c$; the hypothesis then forces $t\ge2$.

By Lemma~\ref{S-lem:three-terminal}, it suffices to prove the bound on a
strongly connected terminal restriction with at least three states.
Cyclicity under $a$ is unchanged by restriction to a closed subset, and
$b$ cannot exchange a point inside that subset with one outside it.
Thus the hypothesis is inherited. If the restriction has at most one
transposition, the cases just treated apply. Hence assume strong
connectivity and $2\le t\le c$. In particular $h$ is moved by $b$.

Use $A(S)=a^{-1}(S)$, $B(S)=b(S)$, and put $v=c-t$.
All operator strings below are in inverse-image order from $\{z\}$.
The initial pairs are
\[
 P_i=A^{t-i}(\{z\})=\{T_i,C_{v+i}\},\qquad 0\le i<t.
\]
Their $A$ edges go down this chain; at $P_0$ they contract to a singleton
if $v>0$ and return to the collision pair if $v=0$.
None is an expansion target. A closed family of such pairs excludes a
triple, since any singleton contraction permits a restart at the same
unique first collision pair. Once two cyclic points touching $C$ are
reached, stop at the first visit to $z$ and apply one more $A$ to expand.

Let $I$ be the tail indices of tail--cycle transpositions, $X$ their
cyclic endpoints, and $k$ the number of tail--tail transpositions.
Then $|I|=|X|=:m$ and $t\ge m+2k$. We have $m\ge1$, since otherwise
both letters permute the invariant cyclic-state set.
If $C_{v+i}\notin X$ for some $i\in I$, then
\[
 B(P_i)=\{b(T_i),C_{v+i}\}
\]
is a cyclic pair. Its first visit to $z$ takes $j\le v+i$ steps,
giving total length $(t-i)+1+j+1\le c+2\le n$.

Otherwise injectivity of $i\mapsto C_{v+i}$ and equality of cardinalities
give $X=\{C_{v+i}:i\in I\}$. All moved cyclic points are then on $C$;
strong connectivity excludes other cycles, so $n=t+c$.
There is a permutation $f$ of $I$ such that
\[
 b(T_i)=C_{v+f(i)}\quad(i\in I).
\]
The permutation $f$ need not be an involution.

Suppose first that $k\ge1$, and put $i_0=\min I$.
If $i_0=0$, choose any tail exchange $(T_p,T_q)$ with $p<q$.
If $i_0>0$, choose one with $p<i_0<q$: such an edge exists because
a path from $C$ to $h$ first enters the tail at an index at least $i_0$,
and its first visit below $i_0$ cannot be an $a$ step or a cross edge.
In either case $q>i_0$ and $p,q\notin I$.
The string $A^{t-p},B$ gives $\{T_q,C_{v+p}\}$, because its cyclic
partner is fixed by $B$. Put $d=q-i_0$.
If $v+p<d$, then $A^{v+p}$ reaches $z$ while the tail index is
$q-(v+p)>i_0\ge0$; one more $A$ expands, at total cost $c+2$.
Otherwise $A^d$ gives $\{T_{i_0},C_w\}$ with
\[
 w=v+p-q+i_0\ge0,\qquad w<v+i_0.
\]
All moved cyclic coordinates are at least $v+i_0$, so $C_w$ is fixed
by $b$. The next $B$ gives the distinct cyclic points
$\{C_{v+f(i_0)},C_w\}$. Applying $A^w$ and expanding costs
\[
 (t-p)+1+d+1+w+1=c+3\le n,
\]
since $t\ge m+2k\ge3$. Equality $v+p=d$ belongs to this latter
branch: a possible intermediate $\{h,z\}$ is expanded only after $B$
has removed the hole.

Now let $k=0$. Then $0\in I$, and
\[
 B(P_i)=\{T_{f^{-1}(i)},C_{v+f(i)}\}\quad(i\in I).
\]
If $f^2$ is the identity, this is $P_{f(i)}$; all other initial pairs
are fixed by $B$. The initial chain is closed without an expansion
target, contradicting synchronization.
Otherwise choose a cycle of $f$ of length at least three, let $k_0$ be
its smallest index, and set $i=f^{-1}(k_0)$, $j=f^{-1}(i)$.
These indices are distinct and satisfy $k_0<i$ and $k_0<j$.
From $P_i$, the next $B$ gives $\{T_j,C_{v+k_0}\}$.
If $v+k_0<j$, the first visit to $z$ occurs with a positive tail index,
and the expansion costs
\[
 (t-i)+1+(v+k_0)+1=c+k_0-i+2\le c+1.
\]
Otherwise $A^j$ gives $\{h,C_w\}$ with
$w=v+k_0-j\ge0$ and $w<v$. Thus $C_w$ is fixed by $b$, whereas $h$
is sent to $C_{v+f(0)}$. After $B$, applying $A^w$ and expanding costs
\[
 (t-i)+1+j+1+w+1=c+k_0-i+3\le c+2\le n.
\]
Here $t\ge m\ge3$. In the boundary $v+k_0=j$, one has $w=0<v$;
the intermediate $\{h,z\}$ is again converted by $B$ before expansion.
This exhausts the cases.
\end{proof}

\begin{Stheorem}\label{S-thm:three-swaps}
Let $\mathcal A=(Q,\{a,b\})$ be synchronizing on $n\ge6$ states.
If $\operatorname{rank}(a)=n-1$ and $b$ consists of exactly three
disjoint transpositions, fixing all remaining states, then
\[
 \Tr(\mathcal A)\le n+2.
\]
The constant $+2$ cannot be replaced uniformly by $+1$.
\end{Stheorem}
\begin{proof}
Lemma~\ref{S-lem:three-terminal} reduces to a strongly connected terminal
restriction. With at most two transpositions use Theorem~\ref{S-thm:sparse};
with none, a synchronizing unary action gives $a^2$.
It remains to prove the bound for a strongly connected automaton with
three transpositions. Its missing image $h$ is moved by $b$.

Write $t,c$ for the tail and receiving-cycle lengths of $a$. If $t>c$,
the word $a^{c+1}$ merges three states. Suppose $t\le c$.
Let $\kappa,\xi,\rho$ count respectively the transpositions with two
tail endpoints, one tail and one cyclic endpoint, and two cyclic endpoints.
Then $\kappa+\xi+\rho=3$, and $\xi\ge1$ because a cyclic state must
be able to reach $h$. The possibilities and their bounds are
\[
\begin{array}{c|c|c|c}
\kappa&\xi&\rho&\hbox{bound}\\\hline
0&3&0&n\\
1&2&0&n\\
2&1&0&n\\
0&1&2&n+2\\
0&2&1&n+2\\
1&1&1&n+2
\end{array}
\]
The first three rows follow from Theorem~\ref{S-thm:no-rr}. In the fourth
row $h$ is the only moved tail point, and Lemma~\ref{S-lem:three-one-tail}
applies, including $t=1$. The fifth row has two moved tail points,
one of them $h$, and is Lemma~\ref{S-lem:three-two-tail}. The last row
is Lemma~\ref{S-lem:three-three-tail}. Their full parameter-uniform proofs
are given in Appendix~\ref{S-app:three-proofs}; none uses a finite
parameter cutoff. These rows exhaust the possible transposition types.
The example below proves the assertion about the constant.
\end{proof}

\begin{Sproposition}\label{S-prop:three-tight}
There is a strongly connected synchronizing seven-state automaton with
exactly three transpositions in $b$ and $\Tr=9$.
\end{Sproposition}
\begin{proof}
On $Q=\{0,\ldots,6\}$ take the transition rows
\[
 a=(0,2,3,4,1,5,0),\qquad b=(1,0,5,6,4,2,3).
\]
The $a$-cycle $(1,2,3,4)$ is joined in both directions by $b$ to each
of $0,5,6$, so the automaton is strongly connected. Direct composition
shows that
\[
 \texttt{abaaabababaabababaaababa}
\]
maps all seven states to $0$, establishing synchronization.
The word $\texttt{abaaababa}$ maps $0,1,6$ to one state.
For the lower bound, the initial inverse pair is $K=\{0,6\}$.
Its reachable pair graph before the first expansion is completely listed
below; $d$ is the distance from $K$.
\[
\begin{array}{c|c|c|c}
 P&d&a^{-1}(P)&b(P)\\\hline
 \{0,6\}&0&\{0,6\}&\{1,3\}\\
 \{1,3\}&1&\{2,4\}&\{0,6\}\\
 \{2,4\}&2&\{1,3\}&\{4,5\}\\
 \{4,5\}&3&\{3,5\}&\{2,4\}\\
 \{3,5\}&4&\{2,5\}&\{2,6\}\\
 \{2,5\}&5&\{1,5\}&\{2,5\}\\
 \{2,6\}&5&\{1\}&\{3,5\}\\
 \{1,5\}&6&\{4,5\}&\{0,2\}\\
 \{0,2\}&7&\{0,1,6\}&\{1,5\}
\end{array}
\]
A singleton contraction allows a shorter restart and cannot improve the
distance to the first target, as in~\eqref{S-eq:sparse-pairdistance}.
It takes one step to reach $K$, at least seven more to reach a target,
and one expansion, so $\Tr\ge9$. The displayed nine-letter word attains it.
\end{proof}

Theorem~\ref{S-thm:no-rr} shows that a value exceeding $n$ in the involution
class requires a transposition between two cyclic states of $a$.
Theorem~\ref{S-thm:three-swaps} does not cover four or more arbitrary
transpositions. In particular, it does not prove a uniform $n+3$ bound
for all involution letters.

\section{One cyclic--cyclic transposition}\label{S-sec:orr}

\begin{Stheorem}\label{S-thm:orr-one-cyclic}
Let $\mathcal A=(Q,\{a,b\})$ be synchronizing on $n\ge3$ states, with
$\operatorname{rank}(a)=n-1$ and $b$ an involution.
If at most one transposition of $b$ has both endpoints cyclic under $a$,
then $\Tr(\mathcal A)\le n+3$. The number of the other transpositions
is unrestricted.
\end{Stheorem}

We give the auxiliary arguments and then prove the theorem. Throughout
these arguments, $n\ge3$, $\operatorname{rank}(a)=n-1$, and $b$ is an
involution. Write the unique tail and its receiving cycle as
\[
 T_0=h\to\cdots\to T_{t-1}\to C_0=z,\qquad
 C=(C_0,\ldots,C_{c-1}).
\]
The other components of the functional graph of $a$, if present, are
cycles. Let $A(S)=a^{-1}(S)$ be the full inverse image and $B(S)=b(S)$.
Sequences of these operators are read in application order from
$\{z\}$; reversing their letters gives a forward word of the same length.
For a pair $P$,
\[
 |A(P)|=2+\boldsymbol1_{z\in P}-\boldsymbol1_{h\in P}.
\]
Every first inverse growth from a singleton gives the collision pair
$K=\{T_{t-1},C_{c-1}\}$. A later singleton contraction permits a restart,
and an empty preimage cannot recover. Thus closure of all reachable
size-preserving pair steps, with only singleton contractions allowed as
exits, excludes every triple word.

If $t>c$, the word $a^{c+1}$ merges the three distinct points
$C_{c-1},T_{t-1},T_{t-c-1}$. Hence the remaining arguments use $t\le c$,
$v=c-t\ge0$, and
\begin{equation}\label{S-eq:orr-chain}
 P_i=A^{t-i}(\{z\})=\{T_i,C_{v+i}\},\qquad 0\le i<t.
\end{equation}
For $i>0$, $A(P_i)=P_{i-1}$. At $P_0$, $A$ contracts to a singleton
when $v>0$ and returns to $K$ when $v=0$. None of these pairs is an
expansion target. From a cyclic pair containing $C_q$, at most $q$
rotations and one expansion suffice: stop at the first visit to $z$.
Equivalently, $A^{q+1}$ has at least three preimages, since it includes
$T_{t-1}$ and the two distinct cyclic preimages of the original pair.

Let $I\subseteq\{0,\ldots,t-1\}$ index the tail--cycle transpositions,
and let $X$ be their cyclic endpoints. Strong connectivity implies $I\ne\varnothing$, since a cycle must reach the tail.
It also forces $b(h)\ne h$, since $h$ has no incoming $a$ edge.
The following cut observation will be used repeatedly.

\begin{Slemma}\label{S-lem:orr-tail-cut}
Assume strong connectivity and at least one tail--tail transposition.
Put $i_0=\min I$. If $i_0=0$, any tail exchange $(T_p,T_q)$ with $p<q$
has $q>i_0$. If $i_0>0$, some such exchange has $p<i_0<q$.
In either case $p,q\notin I$.
\end{Slemma}
\begin{proof}
Only the second assertion needs proof. A path from a cyclic point to $h$
enters the tail at an index in $I$, hence at least $i_0$.
Its first subsequent visit below $i_0$ cannot be a forward $a$ step or
a tail--cycle exchange. It must use a tail exchange across that cut.
Its upper endpoint cannot be $i_0$, already paired with a cyclic point.
\end{proof}

\begin{Slemma}\label{S-lem:orr-low-pair}
Suppose a strongly connected automaton has reached the inverse pair
$\{h,C_\delta\}$, where $\delta\ge0$ and $b$ fixes every $C$ point of
coordinate at most $\delta$. A triple can be obtained in at most
$\delta+B_0+1$ additional steps, where $B_0\le t$.
Consequently, if $B(P_s)=\{T_j,C_{v+k}\}$ with $k<j$, $k\le s$, and
all moved $C$ points have coordinates at least $v$, then
\begin{equation}\label{S-eq:orr-negative}
 \Tr\le c+k-s+B_0+2\le c+t+2.
\end{equation}
\end{Slemma}
\begin{proof}
Reverse a path from its last entry into the tail to its first visit at
$h$, and remove repeated tail vertices. This gives a simple inverse path
from $h$ through tail states to a cyclic point. It has at most $t$ edges,
including its final cross $B$; let $B_0\le t$ count its $B$ steps.
The path starts with $B$ and does not revisit $h$.
Carry $C_\delta$ along it. Every $B$ fixes this coordinate and every $A$
lowers it by one. If it reaches $z$ before the path ends, the tail point
is positive and one further $A$ expands. For $\delta=0$ initially,
first take the initial $B$, which removes $h$ and fixes $z$.
If the path finishes first, the remaining rotations and expansion
complete a cyclic pair. In all cases the $A$ costs total at most
$\delta+1$, proving the first assertion.

For the consequence, start with $A^{t-s},B$. If $v+k<j$, an early
expansion costs $c+k-s+2$. Otherwise $A^j$ gives
$\{h,C_\delta\}$ with $\delta=v+k-j\ge0$ and $\delta<v$.
The cost through the low-pair escape is
$(t-s)+1+j+\delta+B_0+1=c+k-s+B_0+2$.
Equality $v+k=j$ is handled by the first assertion at $\{h,z\}$,
not by an immediate expansion of that pair.
\end{proof}

\begin{Slemma}\label{S-lem:orr-off-alignment}
Assume strong connectivity, synchronization, and exactly one
cyclic--cyclic transposition. If
$X\ne\{C_{v+i}:i\in I\}$, then $\Tr\le n+1$.
\end{Slemma}
\begin{proof}
If $C_{v+i}\notin X$ for $i\in I$, then $B(P_i)$ is cyclic.
If it touches $C$, its cost through expansion is at most
$(t-i)+1+c\le n+1$. Suppose none of these routes applies.
Because the two support sets have the same size but are different,
there is such an index $k$. Both $b(T_k)$ and $b(C_{v+k})$ must be
outside $C$. The latter is possible only at the $C$ endpoint of the
unique cyclic exchange. Thus $k$ is unique, and the support has the form
\[
 b(T_k)=x\notin C,\qquad (C_{v+k},y)\text{ is the cyclic exchange},
 \qquad y\notin C.
\]
The other cross endpoints are exactly $C_{v+i}$ for
$i\in J=I\setminus\{k\}$, and their assignment is a permutation
$f:J\to J$, with $b(T_i)=C_{v+f(i)}$.
All moved $C$ points have coordinates $\{v+i:i\in I\}$.
No further cycle can avoid both $x,y$, since it would have no moved
point and violate strong connectivity.

If a tail exchange exists, use Lemma~\ref{S-lem:orr-tail-cut} and put
$d_0=q-i_0$. The prefix $A^{t-p},B$ gives $\{T_q,C_{v+p}\}$, since
$p\notin I$. If $v+p<d_0$, early expansion costs $c+2$.
Otherwise $A^{d_0}$ gives $\{T_{i_0},C_w\}$ with
$w=v+p-q+i_0\ge0$ and $w<v+i_0$.
This $C$ point is fixed by $B$, which sends the tail point to a cyclic
state. Applying $A^w$ and expanding costs $c+3\le n$, since $t\ge3$.
Equality $v+p=d_0$ belongs to this latter branch.

Hence assume there are no transpositions between tail states, so $0\in I$.
If $f$ has a cycle of length at least three, choose its minimum $q$,
let $i=f^{-1}(q)$ and $j=f^{-1}(i)$, and note $q<i,j$.
Then $B(P_i)=\{T_j,C_{v+q}\}$.
An early $z$ hit costs $c+q-i+2$; otherwise $A^j$ gives
$\{h,C_w\}$ with $0\le w=v+q-j<v$.
The next $B$ fixes $C_w$ and sends $h$ to a cyclic point.
Expansion costs $c+q-i+3\le c+2\le n$.

It remains that $f^2$ is the identity. The initial pairs indexed by $J$
are fixed or exchanged according to $f$, all other initial pairs are
fixed, and $P_k$ is sent to $\{x,y\}$.
If $x,y$ lie on the same cycle of length $d$, then $n=t+c+d$.
If they are not antipodal, some $1\le j\le d-1$ inverse rotations send
$x$ to $y$ and the other point to $w\ne x$.
The next $B$ gives $\{C_{v+k},w\}$, at total cost
\[
 (t-k)+1+j+1+(v+k)+1=c+j+3\le c+d+2\le n+1.
\]
If they are antipodal, their rotations form a matching. $B$ fixes each
pair except $\{x,y\}$, which returns to $P_k$. This matching and the
initial chain are closed without a target, excluding synchronization.

Suppose instead that their cycle lengths are $e,d$, respectively,
on distinct cycles. Then $n=t+c+e+d$.
If $e\nmid d$, rotate by $d$: the $y$ point returns to $y$ and the
$x$ point does not return to $x$. Applying $B$ and completing through
$C_{v+k}$ costs $c+d+3\le n+1$.
If $e=d$, the rotations form a matching whose two marked points occur
simultaneously. It is fixed by $B$ except for return to $P_k$; together
with the initial chain it is a closed no-triple family.

Finally let $e\mid d$ and $e<d$. Strong connectivity forces $k>0$:
for $k=0$, the set consisting of $h$ and the $x$-cycle has no incoming
transition from outside. Thus $0\in J$ and $b(h)=C_{v+f(0)}$.
Rotate $\{x,y\}$ by $e$ and apply $B$, obtaining $\{T_k,y'\}$ with
$y'$ on the $y$-cycle. After $A^k$ and another $B$ this is a cyclic
pair containing $C_{v+f(0)}$; its other point stays external or returns
to $C_{v+k}$. It cannot become $h$, because $x$ is on a different cycle.
The total cost is at most
\[
 (t-k)+1+e+1+k+1+(v+f(0))+1
 =c+e+f(0)+4
 \le c+e+t+3\le n+1,
\]
using $f(0)\le t-1$ and $d>e\ge1$, hence $d\ge2$.
All closed-family assertions concern size-preserving pair steps;
any additional singleton contractions are covered by the restart rule.
\end{proof}

In the aligned case write
\begin{equation}\label{S-eq:orr-alignment}
 X=\{C_{v+i}:i\in I\},\qquad b(T_i)=C_{v+f(i)},\quad f:I\to I
 \text{ a permutation}.
\end{equation}
Let $\theta$ be the involution of
$J=\{0,\ldots,t-1\}\setminus I$ induced by transpositions between tail states, fixing
unmatched points. A single cyclic exchange is then either internal to
$C$, or joins $C$ to one external cycle $D$. Strong connectivity excludes
any other cycle; hence $n=t+c$ or $n=t+c+d$, where $d=|D|$, respectively.
The following aligned lemmas include the remaining dependencies of the proof.

\begin{Slemma}\label{S-lem:orr-high}
Assume strong connectivity, synchronization, and
\eqref{S-eq:orr-alignment}. If every cyclic--cyclic exchange has both
endpoints in $\{C_v,\ldots,C_{c-1}\}$, then $\Tr\le n+2$.
Here the number of cyclic--cyclic transpositions is also unrestricted.
\end{Slemma}
\begin{proof}
These transpositions induce an involution $\sigma$ of $J$ by
$b(C_{v+s})=C_{v+\sigma(s)}$. If $f$ has a cycle of length at least
three, choose its minimum $k$, $s=f^{-1}(k)$ and $j=f^{-1}(s)$.
Then $k<s,j$ and $B(P_s)=\{T_j,C_{v+k}\}$, so
\eqref{S-eq:orr-negative} applies.

Suppose $f^2$ is the identity. If $\theta=\sigma$, $B$ preserves the
initial chain, contradicting synchronization. Otherwise choose a
component of the union of the two matchings which is neither an isolated
point nor a common single edge, and let $k$ be its minimum.
If $\sigma(k)=k$, take $s=k$, so $\theta(s)>k$.
Otherwise take $s=\sigma(k)>k$. Equality $\theta(s)=k$ would make the
component a common single edge, so again $\theta(s)>k$.
Thus $\sigma(s)=k\le s$ and $\theta(s)=j>k$.
Equation~\eqref{S-eq:orr-negative} applies once more.
All moved cyclic points lie on $C$, so strong connectivity gives
$n=t+c$, and the bound is $n+2$.
\end{proof}

\begin{Slemma}\label{S-lem:orr-low-internal}
Assume strong connectivity, synchronization, \eqref{S-eq:orr-alignment},
and that the unique cyclic exchange is
$(C_\alpha,C_\beta)$ with $\alpha<\beta<M=v+\min I$.
Then $\Tr\le n+1$.
\end{Slemma}
\begin{proof}
Here $n=t+c$. Suppose first that a tail exchange exists.
Take $i_0,p,q$ from Lemma~\ref{S-lem:orr-tail-cut} and put $d_0=q-i_0$.
The prefix $A^{t-p},B$ gives $\{T_q,C_u\}$.
If $i_0=0$, then $p\ge1$ and $u=v+p$, since both RR endpoints lie
below $v$. If $i_0>0$, then $0\le u<M$, since $B$ preserves that
lower interval.
If $u<d_0$, early expansion costs $c+2$ in the first case and at most
$c+i_0-p+1$ in the second, both at most $n$.
Otherwise $A^{d_0}$ gives $\{T_{i_0},C_w\}$ with $0\le w<M$.
The next $B$ sends the tail point to a cyclic point of coordinate at
least $M$, and keeps its other point below $M$. Completion therefore
costs at most $M$ more steps, for total
\[
 (t-p)+1+(q-i_0)+1+M=c+q-p+2\le n+1.
\]
Equality $u=d_0$ is included in this latter branch.

If there is no tail exchange, then $0\in I$, $M=v$, and $\beta<v$.
If $f^2$ is the identity, the initial pair chain is closed.
Otherwise choose a cycle minimum $k$, $s=f^{-1}(k)$ and
$j=f^{-1}(s)$, with $k<s,j$.
From $P_s$, $B$ gives $\{T_j,C_{v+k}\}$.
An early expansion costs at most $c+1$.
Otherwise $A^j$ gives $\{h,C_w\}$ with $w<v$.
The next $B$ sends $h$ to a high cyclic point and keeps its other
coordinate below $v$. Its total cost is at most
$c+j-s+2\le n$, since $s\ge1$ and $j\le t-1$.
\end{proof}

\begin{Slemma}\label{S-lem:orr-bridge-return}
Assume strong connectivity, \eqref{S-eq:orr-alignment}, and that the
unique cyclic exchange is
$(C_\alpha,Y)$, where $\alpha<v$ and $Y$ lies on an external cycle
$D$ of length $d$. Suppose the cyclic pair $\{C_{v+g},Y\}$,
$0\le g<t$, is reached by an inverse prefix of length $L$, with
$L+\alpha\le c+2$. Then $\Tr\le n+3$.
If $L+\alpha\le c+1$, the same construction gives $n+2$.
\end{Slemma}
\begin{proof}
Here $n=t+c+d$. If $d\ge v$, finish through $C_{v+g}$.
The total is
\[
 L+v+g+1=(L+\alpha)+(v-\alpha)+g+1\le n+2.
\]
If $d<v$, take $J=d(\lfloor g/d\rfloor+1)$, so $g<J\le g+d$.
After $A^J$, the $D$ point returns to $Y$ and the $C$ coordinate is
$\gamma=v+g-J$, with $0<v-d\le\gamma<v$.
No intermediate $C$ point reaches $z$.
If $\gamma=\alpha$, finish directly.
Otherwise $B$ gives $\{C_\alpha,C_\gamma\}$, because $\gamma$ is below
all cross coordinates and is not $\alpha$.
Completion costs at most $\alpha+1$ more steps, giving
\[
 L+J+1+\alpha+1
 \le c+2+(t-1+d)+2=n+3.
\]
The $B$ step is unused when $\gamma=\alpha$.
Replacing $c+2$ by $c+1$ in the hypothesis improves both bounds by one.
This uses the explicit $J\le g+d$, not a least-common-multiple estimate.
\end{proof}

\begin{proof}[Proof of Theorem~\ref{S-thm:orr-one-cyclic}]
The case without a cyclic--cyclic exchange follows from
Theorem~\ref{S-thm:no-rr}.
Use Lemma~\ref{S-lem:three-terminal} to pass to the terminal strongly
connected restriction; components of size at most two already give
$\Tr\le3$. Cyclicity is unchanged and $b$ cannot cross the component,
so at most one cyclic exchange is retained.
We may therefore assume strong connectivity and exactly one such
exchange. The long-tail case was handled above.
Lemma~\ref{S-lem:orr-off-alignment} permits us to assume
\eqref{S-eq:orr-alignment}.

\emph{The exchange is internal to $C$.}
Write it as $(C_\alpha,C_\beta)$ with $\alpha<\beta$.
If $\alpha\ge v$, Lemma~\ref{S-lem:orr-high} applies.
If $\beta<v$, Lemma~\ref{S-lem:orr-low-internal} applies.
Otherwise $\alpha<v\le\beta=v+s$, where $s\in J$.
Put $j=\theta(s)$. If $j>0$, the prefix $A^{t-s},B$ gives
$\{T_j,C_\alpha\}$.
An earlier $z$ visit expands immediately; otherwise $A^j$ gives
$\{h,C_{\alpha-j}\}$, strictly below $\alpha$, the lowest moved
$C$ coordinate. Lemma~\ref{S-lem:orr-low-pair} gives
\[
 \Tr\le(t-s)+1+j+(\alpha-j)+B_0+1
 \le n+1.
\]
This includes $\alpha=j$.

Only $j=0$ remains, so $s>0$ and the tail exchange is $(h,T_s)$.
Let $i_0=\min I>0$. If $i_0>s$, both RR endpoints lie below
$v+i_0$, and Lemma~\ref{S-lem:orr-low-internal} applies.
Thus $i_0<s$.
For any $i\in I$ with $i<s$, start with $A^t,B$ to get
$\{T_s,C_v\}$ and follow $A$ toward $T_i$.
If $z$ is reached by then, the tail is positive and expansion costs
$c+2$. Otherwise the cyclic coordinate is $w=v-s+i>0$, below $v$.
If $w\ne\alpha$, the next $B$ and completion cost $c+3$.
Consequently only the following configuration can remain:
\[
 I\cap\{1,\ldots,s-1\}=\{i_0\},\qquad
 \alpha=v-s+i_0>0.
\]
All other cross indices exceed $s$; the $\alpha=0$ case was an early
expansion at $T_{i_0}$.

Put $g=f(i_0)$. If $g=i_0$, the sequence
$A^t,B,A^{s-i_0},B$ yields
$\{C_{v+i_0},C_{v+s}\}$ and costs $c+s+3\le n+2$ through expansion.
If $g\ne i_0$ and $f(g)=i_0$, then $g>s$ and
\[
 A^{t-g},B,A^{i_0},B,A^{s-i_0},B
\]
uses $B(P_g)=P_{i_0}$ to reach the final pair
$\{C_{v+g},C_{v+s}\}$. Its cost through expansion is
$c-g+2s+4\le c+s+3\le n+2$.
Otherwise let $r=f^{-1}(i_0)$ and $k=f^{-1}(r)$.
Both exceed $s$. From $P_r$, $B$ gives $\{T_k,C_{v+i_0}\}$.
Follow $A$ toward $T_{i_0}$, expanding if $z$ is reached by then.
Otherwise its cyclic coordinate is $w=v+2i_0-k>0$, below $v+i_0$.
If $w\ne\alpha$, the next $B$ and completion cost
$c+i_0-r+3\le c+2$.
If $w=\alpha$, then $k=s+i_0$; the next $B$ gives
$\{C_{v+g},C_{v+s}\}$, with $g>s$, at total cost
$c+2s-r+3\le c+s+2\le n+1$.
This completes every internal position.

\emph{The exchange is a bridge $(C_\alpha,Y)$ to $D$.}
Write $d=|D|\ge1$, so $n=t+c+d$.
First suppose $\alpha<v$.
If a tail exchange exists, choose $i_0,p,q$ as in
Lemma~\ref{S-lem:orr-tail-cut}. The prefix $A^{t-p},B$ gives
$\{T_q,C_{v+p}\}$, since the bridge misses that $C$ point.
If $v+p<q-i_0$, early expansion costs $c+2$.
Otherwise $A^{q-i_0}$ gives $\{T_{i_0},C_w\}$ with
$w=v+p-q+i_0\ge0$ and $w<v+i_0$.
If $w\ne\alpha$, $B$ fixes this low $C$ point, and completion costs
$c+3$. If $w=\alpha$, $B$ gives $\{C_{v+g},Y\}$,
$g=f(i_0)$, after a prefix of length
$L=(t-p)+1+(q-i_0)+1$.
Here $L+\alpha=c+2$, so Lemma~\ref{S-lem:orr-bridge-return} applies.

If there is no tail exchange, then $0\in I$.
When $f^2$ is the identity, the initial chain is closed because the
bridge endpoint $\alpha<v$ misses it.
Otherwise choose a cycle minimum $k$, $i=f^{-1}(k)$ and
$j=f^{-1}(i)$, with $k<i,j$.
From $P_i$, $B$ gives $\{T_j,C_{v+k}\}$.
An early expansion costs at most $c+1$.
Otherwise $A^j$ gives $\{h,C_w\}$, $w=v+k-j\ge0$, $w<v$.
For $w\ne\alpha$, the next $B$ and completion cost
$c+k-i+3\le c+2$.
For $w=\alpha$, the next $B$ gives $\{C_{v+f(0)},Y\}$ after a prefix
of length $L=(t-i)+1+j+1$.
Now $L+\alpha=c+k-i+2\le c+1$;
Lemma~\ref{S-lem:orr-bridge-return} even gives $n+2$.
Thus all low bridges are covered.

It remains that $\alpha=v+s$, $s\in J$.
All moved $C$ coordinates are at least $v$.
A cycle of $f$ of length at least three gives
\eqref{S-eq:orr-negative}, as in Lemma~\ref{S-lem:orr-high}, and hence
$\Tr\le c+t+2\le n+1$.
Assume $f^2$ is the identity.
If a tail exchange $(T_p,T_q)$, $p<q$, has $p\ne s$, then
$B(P_p)=\{T_q,C_{v+p}\}$ and \eqref{S-eq:orr-negative} applies again.
If there is no tail exchange, $s\ge1$ and $T_s$ is fixed by $B$.
The sequence $A^{t-s},B,A^s,B$ then gives a cyclic pair containing
$C_{v+f(0)}$, for cost at most $c+f(0)+3\le c+t+2\le n+1$.
Hence the only remaining tail exchange is $(T_s,T_j)$ with $s<j$.

If $s=0$, strong connectivity supplies a cross index $i<j$.
Put $g=f(i)$. The sequence
\[
 A^{t-g},B,A^i,B,A^{j-i},B
\]
first uses $B(P_g)=P_i$ and reaches $P_0$.
The following $B$ gives $\{T_j,Y\}$.
At the last $B$, the tail point goes to $C_{v+g}$ and the $D$ point
stays cyclic, even if it returns through the bridge to $C_v$.
The cost through completion is at most
$c+j+4\le c+t+3\le n+2$.

Finally let $s>0$, so $h$ is crossed, and put $g=f(0)$.
It cannot equal $s$ or $j$.
If $g<s$, start at $P_s$, use $B,A^j,B$, and finish through $C_{v+g}$.
The cost is $c+j-s+g+3\le c+j+3\le n+1$.
If $s<g<j$, use $B,A^{j-g},B$ from $P_s$.
Since $f(g)=0$, its final cyclic pair contains $C_v$;
completion costs $c+j-s-g+3\le n+1$.

Suppose $g>j$ and put $\Delta=j-s>0$.
If $\Delta\in I$, use the alternative source $P_s$ and take
$B,A^s,B$.
The tail point is $T_\Delta$ before the last $B$, which sends it to
$C_{v+f(\Delta)}$; its $D$ partner remains cyclic.
The total cost is at most
\[
 (t-s)+1+s+1+(v+f(\Delta))+1
 =c+f(\Delta)+3\le c+t+2\le n+1.
\]
If $\Delta\notin I$, start instead at $P_j$ and use $B,A^s,B$.
Before the last $B$ the pair is $\{h,C_{v+\Delta}\}$.
For $\Delta\ne s$ that $C$ point is fixed and completion costs $c+3$.
For $\Delta=s$ it goes to $Y$, while $h$ goes to $C_{v+g}$;
the cost is $c+g-\Delta+3\le c+t+1\le n$.
These cases exhaust the high bridge.

Every aligned internal exchange or high bridge has bound at most $n+2$,
and every low bridge has bound at most $n+3$.
Together with the off-aligned case and the terminal restriction,
this proves the theorem.
\end{proof}

\section{One receiving cycle with unrestricted transpositions on the cycle}
\label{S-sec:internal-involutions}

We retain the notation of Section~\ref{S-sec:orr}: the singular letter has
tail $T_0=h,\ldots,T_{t-1}$ entering $C_0=z$ on a cycle of length $c$;
$A=a^{-1}$ denotes full inverse image and $B=b$. Instruction sequences
are applied from $\{z\}$ in the written order. Reversing their letters
gives ordinary forward words of the same length. We call transpositions
tail--tail, tail--cycle and cyclic--cyclic transpositions, abbreviated TT, TC
and RR only in the descriptions of cases below.

In this section $a$ has exactly one cycle, so $n=t+c$. If $t>c$,
$a^{c+1}$ already merges three states. Otherwise put $v=c-t$ and use
the initial pairs $P_i=\{T_i,C_{v+i}\}$ of \eqref{S-eq:orr-chain}.
Two distinct points of this one cycle have minimum coordinate at most
$c-2$ and expand within at most $c-1$ inverse $a$-steps. In particular,
off-alignment of TC support gives $\Tr\le n$: for some crossed index $i$,
$C_{v+i}$ is not a TC endpoint, and $B(P_i)$ is cyclic, at total cost
at most $(t-i)+1+(c-1)=n-i$. This improvement uses the single-cycle
hypothesis; it is not the external-cycle estimate of
Lemma~\ref{S-lem:orr-off-alignment}.

\subsection{Crossing the missing image directly to the cycle}

\begin{Stheorem}\label{S-thm:internal-crossed-hole}
Let $\mathcal A$ be synchronizing on $n\ge3$ states, with $\operatorname{rank}(a)=n-1$,
exactly one $a$-cycle and an involution $b$. If $b$ exchanges the missing image
$h$ with a cyclic state, then $\Tr\le n$. There is no bound on the number
or placement of the other transpositions.
\end{Stheorem}
\begin{proof}
The $h$-cross makes the automaton strongly connected: from the cycle its
partner reaches $h$, and $a$ traverses the entire tail and cycle. Reduce to
$t\le c$ and aligned support as above. Let $I$ be the crossed indices,
so $0\in I$, and write $b(T_i)=C_{v+f(i)}$ for a permutation $f$ of $I$.
Let $J=\{0,\ldots,t-1\}\setminus I$ and let $\theta$ be the tail
involution on $J$, fixing unmatched points. In particular $0\notin J$.

First suppose an RR transposition crosses the cut $v$. Among
$(C_\alpha,C_{v+s})$, $\alpha<v$, choose the least lower endpoint
$\alpha$. Then $s\in J$, $s\ge1$, and $j=\theta(s)\ge1$.
The instructions $A^{t-s},B$ reach $\{T_j,C_\alpha\}$.
If $\alpha<j$, an early expansion costs
$(t-s)+1+\alpha+1\le c-s+1\le n$.
Otherwise $A^j$ reaches $\{h,C_w\}$ with $w=\alpha-j<\alpha$.
By the minimal choice of $\alpha$, $b$ keeps $C_w$ below $v$: a jump
above the cut would be an RR with a smaller lower endpoint. The next
$B$ sends $h$ to a high cyclic point and the other point to a low one.
Completing in at most $v$ additional $A$ steps costs
\[
(t-s)+1+j+1+v=c+j-s+2\le n.
\]
This also handles $\alpha=j$: at $\{h,z\}$ take $B$ before counting
growth. The cut-crossing case necessarily has $v>0$.

If no RR crosses $v$, the lower interval is $B$-invariant and the high
cyclic exchanges give an involution $\sigma$ on $J$. If $f$ has a cycle
of length at least three, choose its least index $k$, put
$s=f^{-1}(k)$ and $j=f^{-1}(s)$; then $k<s,j$ and
$B(P_s)=\{T_j,C_{v+k}\}$. If instead $f^2$ is the identity and
$\theta\ne\sigma$, use the least vertex $k$ of a component of their
matching union that is neither an isolated vertex nor a common edge.
Take $s=k$ when $\sigma(k)=k$, and $s=\sigma(k)$ otherwise. Then
$\sigma(s)=k\le s$ and $j=\theta(s)>k$, giving the same pair formula.
Indeed $\theta(s)=k$ in the second case would isolate a common edge.

For either pair, if $v+k<j$, early expansion costs $c+k-s+2$.
This is at most $c+1$ in the $f$-cycle case. In the matching case it is
at most $c+2$, and at least two vertices of $J$ together with $0\in I$
force $t\ge3$, so it is at most $n$. Otherwise $A^j$ reaches
$\{h,C_{v+k-j}\}$ below $v$. Taking $B$ moves $h$ high and keeps the other
point low; completion has cost $c+j-s+2\le n$, since $s\ge1$ and
$j\le t-1$. When $v=0$ the early branch is forced by $k<j$.

Finally if $f^2$ is the identity and $\theta=\sigma$, then $B$ permutes
the initial pairs: it acts by $f$ on $I$ and by $\theta$ on $J$. Together
with their $A$-steps this is a closed no-growth family, contradicting
synchronization by the restart argument of Section~\ref{S-sec:orr}.
Every case has now been covered without counting RR transpositions.
\end{proof}

\subsection{Tails of length at most three}

\begin{Stheorem}\label{S-thm:internal-short-tail}
If a synchronizing automaton on $n\ge3$ states has $\operatorname{rank}(a)=n-1$, exactly
one $a$-cycle, an $a$-tail of length at most three, and an involution $b$,
then $\Tr\le n+1$. The bound is attained for every $n\ge12$.
\end{Stheorem}
\begin{proof}[Proof of the upper bound]
Use the terminal-component reduction of Lemma~\ref{S-lem:three-terminal}.
The restriction retains the sole cycle and a suffix of the tail, so the
tail length does not increase. A nontrivial synchronizing restriction
cannot have both letters permutations. Components of size at most two
are covered by that lemma's $\Tr\le3$ alternative. Thus assume strong
connectivity and $t\le c$.

If $h$ is crossed, Theorem~\ref{S-thm:internal-crossed-hole} applies.
If $h$ is fixed by $b$, it has no incoming $a$- or $b$-edge from outside,
contradicting strong connectivity. A TT through $h$ uses the entire tail
when $t=2$, leaving no way from the cycle into it. For $t=3$ the only
possible remaining configuration is
\[
b(T_0)=T_2,\qquad b(T_1)=C_{c-2},
\]
after off-alignment has been removed. The alternative TT $(T_0,T_1)$
leaves those two points unreachable from a cross at $T_2$. The initial
pairs are $K=P_2$, $P_1$ and $P_0$, at costs 1,2,3; $B$ fixes $P_1$.

For $c=3$, the unused cyclic points are 0,2. Swapping them closes the
initial pairs; otherwise $aba^3$ works. For $c=4$, a swap (1,3) closes
them; a swap (0,1) permits $aba^3$; with no swap or with (0,3),
$a^2ba^3$ works. These exhaust the involutions on the unused points.
Assume $c\ge5$.

Let $b(C_{c-3})=C_u$. The value $u=c-1$ closes the initial pair
family, and $u=c-2$ is unavailable. Hence $0\le u\le c-3$.
From $P_0$, $B$ gives $\{T_2,C_u\}$; if $u\le1$, completion costs
at most six. Otherwise $A,B$ give at cost six a cyclic pair
\[
S=\{C_{c-2},C_k\},\qquad C_k=b(C_{u-1}).
\]
Here $k\ne c-2,c-3$: the latter would imply $u-1=u$ by involutivity.
If $k\le c-4$, direct completion costs at most $c+3=n$.
Only the lock $b(C_{u-1})=C_{c-1}$ remains, so
$S=\{C_{c-2},C_{c-1}\}$ is available at cost six.

The second mixed route $A^3,B,A^2,B$ reaches $\{T_2,C_r\}$ at
cost seven, where $b(C_{u-2})=C_r$. The already fixed reverse images
give
\[
0\le r\le c-4,\qquad r\ne u,u-1.
\]
If $r\le1$, direct expansion costs at most nine, within $c+4$ since
$r\le1$ in this branch requires only $c\ge5$. Otherwise $r\ge2$.
For $r>u$, use the retained $S$ and $A^{c-2-r},B$. The cyclic output
contains $C_{u-2}$; its other point remains cyclic because
$r+1\le c-3$. The total with completion is
\[
6+(c-2-r)+1+(u-1)=c+u-r+4\le c+3.
\]
For $r\le u-4$, use $S$ and $A^{c-u},B$. The output is
$\{C_r,C_{c-1}\}$ and the total is $c+r-u+8\le c+4$.

Only $r=u-2$ or $u-3$ remains. Continuing the second mixed route with $A,B$
gives $\{C_{c-2},C_{k'}\}$ at cost nine, with
$b(C_{r-1})=C_{k'}$. The values $k'=c-2,c-1,c-3$ are excluded by
the crossed fibre and the two fixed reverse images. If $k'\le c-6$,
completion costs at most $c+4$. For $k'=c-4$, use $S$ and $A^2,B$;
one output is $C_{r-1}$ and the other is cyclic. Completion costs at
most $r+9$, which is within $c+4$ in both cases for $r$. For $k'=c-5$,
use $S$ and $A^3,B$, costing at most $r+10$. If $r=u-3$ this is at most
$c+4$. For $r=u-2$ it could exceed $c+4$ only when $u=c-3$. But then
$r=c-5$ is fixed by $b$, whereas $b(C_{r-1})=C_r$, violating injectivity.
Thus this last case also has $u\le c-4$ and the desired budget.

All displayed rotations have nonnegative length and do not wrap before
the reserved expansion. Their $B$-inputs avoid the unique crossed cyclic
coordinate. An earlier visit to $z$ has a positive tail companion and
only shortens the word. Hence $\Tr\le c+4=n+1$.
\end{proof}

\subsection{A sharp family and its labeling of state pairs}

\begin{Sproposition}\label{S-prop:two-rr-exact-family}
For $c\ge9$, let $a$ be the path $T_0\to T_1\to T_2\to C_0$ followed
by the $c$-cycle $C$, and put
\[
b=(T_0\ T_2)(T_1\ C_{c-2})(C_0\ C_1)(C_{c-4}\ C_{c-1}).
\]
This automaton is strongly connected and synchronizing, and
$\Tr=c+4=n+1$.
\end{Sproposition}
\begin{proof}
All four swaps are disjoint. Strong connectivity follows from the cross
to $T_1$, then $a$ to $T_2$, $b$ to $T_0$, and $a$ through the tail and cycle.
The inverse instruction sequence
\[
A^3,B,A^2,B,A,B,A^{c-5}
\]
has successive pairs $\{T_0,C_{c-3}\}$, $\{T_2,C_{c-3}\}$,
$\{T_0,C_{c-5}\}$, $\{T_2,C_{c-5}\}$, $\{T_1,C_{c-6}\}$,
$\{C_{c-2},C_{c-6}\}$. Its final set is
$\{T_2,C_3,C_{c-1}\}$. Thus
$W=a^{c-5}baba^2ba^3$ is a triple word of length $c+4$.

For the lower bound use a potential on every unordered pair. On cyclic
pairs set $\Phi(\{C_i,C_j\})=\min(i,j)+1$. On tail pairs set
\[
\Phi(\{T_0,T_1\})=c-1,\quad
\Phi(\{T_1,T_2\})=c,\quad\Phi(\{T_0,T_2\})=c+3.
\]
For mixed pairs write $\Phi_i(j)=\Phi(\{T_i,C_j\})$ and use
the following complete table; the bulk is nonempty even when $c=9$.
\begin{center}\small
\begin{tabular}{@{}c|rrrrrrrrr@{}}
$i\backslash j$&0&1&2&3&$4\le j\le c-5$&$c-4$&$c-3$&$c-2$&$c-1$\\\hline
0&3&2&4&5&$j+3$&$c+4$&$c+1$&$c+1$&$c-1$\\
1&1&2&3&5&$j+2$&$c-1$&$c-1$&$c+2$&$c-2$\\
2&1&2&3&4&$j+2$&$c-2$&$c$&$c$&$c+3$
\end{tabular}
\end{center}
We verify $\Phi(P)\le1+\Phi(Q)$ on every pair-preserving inverse
edge and $\Phi(P)\le1$ on every immediate growth edge.

For $A$, cyclic pairs rotate down by one or expand from a pair of value
one containing $C_0$. The mixed pairs with $i\ge1$, $j=0$ expand and
have value one. For $i\ge1$, $j\ge1$, the target is
$\{T_{i-1},C_{j-1}\}$. The margins $1+\Phi_{i-1}(j-1)-\Phi_i(j)$
are
\begin{center}\small
\begin{tabular}{@{}c|rrrrrrrrr@{}}
$i\backslash j$&1&2&3&4&$5\le j\le c-5$&$c-4$&$c-3$&$c-2$&$c-1$\\\hline
1&2&0&0&0&1&0&6&0&4\\
2&0&0&0&0&0&0&0&0&0
\end{tabular}
\end{center}
All are nonnegative. At $c=9$ the middle interval is empty and the
columns still do not overlap. The pair $\{T_0,C_0\}$ returns to $K$
and satisfies $3\le1+(c+3)$; other mixed pairs with $T_0$ contract.
The only pair-preserving two-tail edge is
$\{T_1,T_2\}\to\{T_0,T_1\}$, with equality of its bound.

For $B$ it suffices to check absolute potential differences at most one.
Let $\sigma=(0\ 1)(c-4\ c-1)$ on cyclic indices other than the
crossed index $c-2$. The mixed $i=0,2$ rows are exchanged with $j$ replaced
by $\sigma(j)$, and their values differ by one. At $j=c-2$ their
partners are the two-tail pairs with values $c$ and $c-1$, again differing
by one. The pair $\{T_0,T_2\}$ and mixed pair $\{T_1,C_{c-2}\}$
are fixed. Every other $i=1$ mixed pair is exchanged with
$\{C_{c-2},C_{\sigma(j)}\}$; the absolute differences for
$j=0,1,2,3$, the bulk, $c-4$, $c-3$, $c-1$ are respectively
$1,1,0,1,1,0,1,1$. Finally, for two cyclic points not containing $c-2$,
the low swap changes the minimum by at most one. The high swap leaves
any other coordinate at most $c-5$ as the minimum; its endpoint pair is
fixed, and the only remaining other coordinate is $c-3$, again changing
the minimum by one. This exhausts $B$.

Every first singleton expansion is the initial $A$-step to $K$. A later
contraction permits a shorter restart at the next such expansion, so
a shortest successful inverse word has a pair-preserving continuation
from $K$ to its first growth. Since $\Phi(K)=c+3$, that continuation
needs at least $c+3$ letters, counting the final expansion. Together with
the first $A$ this proves $\Tr\ge c+4$.

To prove synchronization, $a$ rotates $C$ and $ba^c$ acts on $C$ as
\[
\sigma=(C_0\ C_1)(C_{c-4}\ C_{c-1}).
\]
The crossed
point returns from $T_1$ to $C_{c-2}$ after $a^c$.
Rotate and orient any cyclic pair as $\{C_0,C_d\}$,
$1\le d\le\lfloor c/2\rfloor$. If $d\ge2$, then
$c-4>\lfloor c/2\rfloor$, so applying $\sigma$ and the backward
rotation sends it to $\{C_0,C_{d-1}\}$. Every such operation and its
inverse is available as a word permuting $C$. Thus any cyclic pair can
be sent to $\{C_3,C_{c-1}\}$, which $W$ merges. Start by $a^3$ to
enter $C$; repeatedly apply these permutations, $W$, and $a^3$. The image
cardinality strictly decreases until it is one. This proves synchronization
for every $c\ge9$ and completes the sharpness assertion above.
\end{proof}

\section{Two transpositions on cycles with restricted tail--cycle endpoints}\label{S-sec:two-rr}

The next results keep exactly one $a$-cycle, but permit arbitrary tail length
and arbitrarily many TT transpositions. They use the initial pairs and original
letter budgets of Section~\ref{S-sec:orr}, especially the low-pair escape
\SuppPrefix Lemma~\ref{S-lem:orr-low-pair} and the arbitrary-RR all-high
Lemma~\ref{S-lem:orr-high}.

\begin{Stheorem}\label{S-thm:one-cross-two-rr}
Let a synchronizing binary automaton on $n\ge3$ states have $\operatorname{rank}(a)=n-1$,
exactly one $a$-cycle, an involution $b$ and at most two RR transpositions.
If it has exactly one TC transposition, then $\Tr\le n+2$.
\end{Stheorem}

\begin{Stheorem}\label{S-thm:diagonal-two-rr}
Let a synchronizing binary automaton on $n\ge3$ states have
$\operatorname{rank}(a)=n-1$, exactly one $a$-cycle, an involution $b$,
and at most two RR transpositions. Suppose $t\le c$ and every TC transposition
is $(T_i,C_{v+i})$, where $v=c-t$. Then $\Tr\le n+3$, with no bound
on the numbers of TC or TT transpositions.
\end{Stheorem}

\begin{proof}[Proof of the two theorems]
Terminal restriction retains $C$ and a tail suffix. Every TC endpoint on $C$
forces its paired tail point into the closed component, so the unique TC
is retained. Diagonal alignment is also inherited: deleting $\ell$ initial
tail points replaces $t,i$ by $t-\ell,i-\ell$ and leaves
$(c-t)+i$ unchanged. Use Lemma~\ref{S-lem:three-terminal} for components
of size at most two, and otherwise assume strong connectivity. The case
$t>c$ in the first theorem is already covered by $a^{c+1}$.

If $h$ is crossed, Theorem~\ref{S-thm:internal-crossed-hole} gives $n$.
If the sole TC is off-aligned, the cyclic-pair argument at the start of
Section~\ref{S-sec:internal-involutions} gives $n$ as well. Thus in either
remaining case the minimum cross index $i$ is positive and
$b(T_i)=C_M$, where $M=v+i$. No cyclic point below $M$ is crossed.
Strong connectivity supplies a TT exchange $(T_p,T_q)$ with $p<i<q$.

\emph{Nonmirrored upper TT edges.} Consider any TT $(T_p,T_q)$ with
$p<q$ and $q>i$, including edges wholly above $i$. Suppose its cyclic mirror
$(C_{v+p},C_{v+q})$ is absent. At cost $t-p$, reach $P_p$, then apply $B$
to obtain $\{T_q,C_u\}$. An $A$-run of length $q-i$ either expands
earlier with positive tail coordinate, or reaches
$\{T_i,C_w\}$ with $w=u-q+i\ge0$. The value $w=M$ is exactly the
excluded mirror. If $C_w$ is not crossed, $B$ gives two cyclic points
containing $C_M$; completion gives the budget
\begin{equation}\label{S-eq:two-rr-nonmirror-cost}
(t-p)+1+(q-i)+1+(v+i)+1=c+q-p+3\le n+2.
\end{equation}
This settles the one-TC case at this edge.

In the diagonal case the only further possibility is $w=v+k$ with
$k\in I$, $k>i$. $B$ gives $\{T_k,C_M\}$. After $A^{k-i}$, unless
growth occurred earlier, the other coordinate is $v+2i-k<M$ at $T_i$.
The next $B$ is cyclic and completion costs at most
\[
c+q+k-i-p+4.
\]
But $u=v+k+q-i\le c-1$ was an actual nonwrapping cyclic coordinate,
so $q+k-i\le t-1$, and the total is at most $n+3$. Earlier expansion
in this extra run is already within \eqref{S-eq:two-rr-nonmirror-cost}.

\emph{The remaining mirror.} We may now assume that every TT with
upper endpoint greater than $i$ has its cyclic mirror. Two such edges
consume both RR transpositions and put them entirely at coordinates at least $v$;
Lemma~\ref{S-lem:orr-high} gives $n+2$. Otherwise the unique such edge is
\[
(T_p,T_s),\quad p<i<s,\qquad(C_{v+p},C_{v+s})\text{ its mirror}.
\]
Every other TT is below $i$. An absent or wholly high second RR again
uses the all-high lemma. The mirror provides the cheap entrance
\begin{equation}\label{S-eq:mirror-cheap-entrance}
A^{t-s},B,A^{p-r}:\{z\}\longmapsto P_r,
\qquad r\le p,
\end{equation}
of length $t-s+1+p-r$. If $p>0$, first descent below $p$ from the cycle
forces another TT $(T_r,T_q)$ with $r<p<q<i$: the mirror reaches $p$ but
not below it, and all other upper tail endpoints are below $i$.

\emph{Second RR below v.} The lower cyclic interval is $B$-invariant.
For $p>0$, take the cheap entrance \eqref{S-eq:mirror-cheap-entrance}, then
\[
B,A^{q-p},B,A^{s-i},B.
\]
The first $B$ gives $\{T_q,C_{v+r}\}$. Before the middle $B$ its cyclic
coordinate is $v+r-q+p<v+p$. That $B$ keeps it below $M$: high points
below $v+p$ are fixed, and low points stay below $v$. After the next
$A$-block it is still below $M$, so the last $B$ gives a cyclic pair
containing $C_M$. The total with completion is
\begin{equation}\label{S-eq:mirror-five-block-cost}
c+q-r+5\le n+2,
\end{equation}
because $q<i<s<t$ implies $q\le t-3$. Any earlier $z$-visit has positive
tail coordinate and shortens the construction.

For $p=0$, if no other TT exists, the initial pairs are closed: $B$ exchanges
$P_0,P_s$ and fixes the rest, including every diagonal crossed pair.
Otherwise choose TT $(T_r,T_q)$ with $1\le r<q<i$ and use
\[
A^{t-r},B,A^q,B,A^{s-i},B.
\]
The pair after $A^q$ is $\{h,C_{v+r-q}\}$ below $v$, unless expansion
occurred earlier. At equality with $z$ take $B$ first, moving $h$ to $T_s$.
The cyclic point stays below $v$ through the remaining instructions;
completion needs at most $v$ $A$-steps. The total is
$c+q-r+s-i+3\le c+s-r+2\le n$.

\emph{Second RR crossing v.} Write it as
$(C_\alpha,C_{v+j})$, $\alpha<v$, with $j$ distinct from $p,s$ and all
cross indices. Let $\ell=\theta(j)$ be its tail partner. If $j>i$, then
$T_j$ is fixed, since a second TT above $i$ was excluded. From $P_j$,
the route $B,A^{j-i},B$ reaches a cyclic pair containing $C_M$, or
expands sooner, with total at most $c+3$.

If $j<i$ and $\ell>0$, use $A^{t-j},B,A^\ell$ to reach
$\{h,C_{\alpha-\ell}\}$, or expand earlier. Every cyclic point
at most $\alpha-\ell$ is fixed, since $\alpha$ is the unique low
moved point. Lemma~\ref{S-lem:orr-low-pair} gives total
\[
t-j+\alpha+B_0+2\le c-j+B_0+1\le n+1,
\qquad B_0\le t.
\]
The escape may leave through any TC edge, so this also handles the
diagonal case with additional crosses.

Finally $\ell=0$ implies $p>0$ and TT $(T_0,T_j)$ below $i$. The cut
exchange $(T_r,T_q)$ above has $r\ne j$. Use the cheap entrance and
the same five blocks as in \eqref{S-eq:mirror-five-block-cost}. Before
the middle $B$ the cyclic coordinate is below $v+p$. $B$ either fixes it,
or exchanges $C_\alpha$ and $C_{v+j}$, both below $M$. The last $B$ is
therefore again cyclic, and the bound \eqref{S-eq:mirror-five-block-cost}
applies. These cases exhaust the second RR position. Every possible
extra crossed input was explicitly treated in the nonmirror case or
avoided by staying below $M$. The only budget $n+3$ arose in the extra
diagonal return, proving both assertions.
\end{proof}

\subsection{Two crossed indices with the transposed assignment}

\begin{Stheorem}\label{S-thm:two-crosses-two-rr}
A synchronizing automaton on $n\ge3$ states with $\operatorname{rank}(a)=n-1$,
exactly one $a$-cycle, an involution $b$, exactly two TC transpositions and at
most two RR transpositions satisfies $\Tr\le n+3$. Tail length and TT
support are unrestricted.
\end{Stheorem}
\begin{proof}
The terminal reduction preserves both TC transpositions, since their cyclic
endpoints lie in the closed receiving cycle. It also preserves the
aligned coordinates. As before assume strong connectivity, $t\le c$,
aligned support and $h$ on a TT exchange. Diagonal assignment is covered
by Theorem~\ref{S-thm:diagonal-two-rr}. The remaining case is
\[
I=\{i,j\},\quad0<i<j<t,\qquad
b(T_i)=C_{v+j},\quad b(T_j)=C_{v+i}.
\]
Put $d=j-i$ and $M=v+i$. The relation $B(P_j)=P_i$ gives an entrance
to any $P_p$, $p<i$, of length $t-j+1+i-p$.

\emph{A nonmirrored straddler.} Strong connectivity supplies a TT
$(T_p,T_q)$ with $p<i<q$. If it lacks its cyclic mirror, use that
cheap entrance followed by $B,A^{q-i}$. Unless there is earlier growth,
the pair is $\{T_i,C_w\}$, $w=u-q+i$, where $b(C_{v+p})=C_u$.
If $C_w$ is not crossed, $B$ gives a cyclic pair containing $C_{v+j}$
and the total is
\[
c+q-p+4\le n+3.
\]
The value $w=v+i$ is the excluded mirror. The only remaining lock is
$w=v+j$, or $u=v+k$ with $k=q+j-i>q,j$ and $k\le t-1$.
Its RR transposition is $(C_{v+p},C_{v+k})$; $B$ fixes the mixed pair
$\{T_i,C_{v+j}\}$, so this is not growth.

Let $\ell=\theta(k)$. If $\ell\ge q$, start from $P_k$; $B$ gives
$\{T_\ell,C_{v+p}\}$. Following $A$ toward $T_i$ either expands or
leaves a cyclic coordinate below $M$, and $B$ completes to a cyclic pair
containing $C_{v+j}$. The total is $c+\ell-q+3\le n+2$.
If $\ell<q$, it is distinct from $p,k,i,j$. Unless the other RR is
absent or wholly high, a case already handled by
Lemma~\ref{S-lem:orr-high}, every high cyclic point off the displayed
exchange is fixed or sent below $v$. Start from $P_\ell$; $B$ sends its
tail point to $T_k$ and its cyclic coordinate to at most $v+\ell$.
After $A^{k-j}$ the coordinate at $T_j$ is at most
$v+\ell-k+j=v+\ell-q+i<M$. The next $B$ gives a cyclic pair
containing $C_M$, with total $c+q-\ell+3\le n+2$.
Thus every nonmirrored TT across $i$ is settled; no assertion has been
made about mirrors of TT edges wholly above $i$.

\emph{The remaining mirrored cut.} If two TT edges across $i$ exist,
their mirrors exhaust both RR and the all-high lemma applies. Otherwise
there is exactly one straddler $(T_p,T_s)$, $p<i<s$, and its cyclic
mirror. Every other TT is wholly below or wholly above $i$. An absent or
high second RR is again settled. If $p>0$, the first descent below $p$
supplies $(T_r,T_q)$ with $r<p<q<i$; use the cheap entrance
\eqref{S-eq:mirror-cheap-entrance}.

If the second RR is wholly below $v$ and $p>0$, the five blocks
$B,A^{q-p},B,A^{s-i},B$ keep the cyclic partner below $M$ before the
last $B$, exactly as in the preceding proof. Completing through
$C_{v+j}$ now gives
\begin{equation}\label{S-eq:two-cross-mirror-budget}
c+q-r+d+5\le n+3,
\end{equation}
since $q+d\le j-1\le t-2$. For $p=0$, a lower TT with
$1\le r<q<i$ gives the same low route as before, of cost
$c+q-r+s-i+3\le n$. If there is no lower TT but there is an upper
TT $(T_r,T_q)$, $i<r<q$, start from $P_r$ and use $B,A^{q-i}$.
At $T_i$ its cyclic coordinate is $v+k$ with $k=r-q+i<i$, unless it
expanded earlier. After $B$, the cases are:
\begin{itemize}
\item $k>0$: the cyclic point is fixed and completion costs $c+3$;
\item $k=0$: the pair has cyclic indices $j,s$, so its total is
$c+\min(j,s)+3\le n+1$, using $j\ne s$;
\item $k<0$: the point stays below $v$, and total cost is
$c+q-r-i+2\le n$.
\end{itemize}
If there is no other TT, $B$ exchanges $P_0,P_s$, exchanges $P_i,P_j$,
and fixes the other initial pairs, contradicting synchronization.

\emph{Second RR crossing v.} Write it as
$(C_\alpha,C_{v+r})$, $\alpha<v$, with $r$ distinct from $p,s,i,j$,
and let $\ell=\theta(r)$. The indices $r$ and $\ell$ lie on the same
side of $i$: otherwise there would be another straddler whose mirror is
unavailable. For $r,\ell>i$, if $\ell\le r$, the route from $P_r$
is $B,A^{\ell-i},B$, completing through $C_{v+j}$ at cost
$c+\ell-r+d+3\le n+1$. If $\ell>r$ and $\ell>j$, stop instead
at $T_j$, giving $c+\ell-r-d+3\le n+2$.
The remaining case is $i<r<\ell<j$. Start from $P_\ell$, whose
cyclic point is fixed by $B$. The route $B,A^{r-i}$ reaches
\[
\{T_i,C_{v+x}\},\qquad x=i+\ell-r,\quad i<x<j.
\]
There is no TC index strictly between $i,j$. Its next $B$ fixes or lowers
that cyclic point, possibly via the mirror's upper endpoint or the
RR endpoint $r$. Completion through the resulting coordinate at most
$v+x$ costs exactly at most $c+3$ in total.

For $r,\ell<i$ with $\ell>0$, $B$ at $P_r$ followed by $A^\ell$
reaches $\{h,C_{\alpha-\ell}\}$, or expands earlier. All coordinates
up to $\alpha-\ell$ are fixed, so the low-pair escape gives
$t-r+\alpha+B_0+2\le n+1$. If $\ell=0$, then $p>0$ and a cut
exchange $(T_u,T_q)$ with $u<p<q<i$ has $u\ne r$. Use its cheap
entrance and $B,A^{q-p},B,A^{s-i},B$. Before the middle $B$ the
coordinate is below $v+p$; that $B$ fixes it or exchanges $C_\alpha$
with $C_{v+r}$, both below $M$. The final $B$ is cyclic and the bound
\eqref{S-eq:two-cross-mirror-budget}, with $u$ in place of $r$, applies.

These cases exhaust both crosses and both RR positions. Every $A$-run
ends at a positive tail index unless a low-pair escape is explicitly
used; in that escape $\{h,z\}$ is followed by $B$ before growth is
claimed. All earlier $z$-visits therefore shorten a valid input-letter
word. This proves the theorem.
\end{proof}

Combining the two theorems, the case of at most two TC transpositions has
$\Tr\le n+3$, with $n+2$ for exactly one. With no TC transposition,
synchronization forces the sole $a$-cycle to have size one; then the
long-tail power already suffices. Diagonal assignment permits any number
of crosses. The two-cross proof uses the absence of a cross index between
$i$ and $j$, and is not a proof for arbitrary nonidentity assignments on
three or more cross indices.

\section{Permutations between paired tail and cycle positions}\label{S-sec:aligned-permutations}

The number of cyclic--cyclic transpositions can be unrestricted when the full
tail--cycle support is aligned on the receiving cycle. The relevant
condition is that the permutation of crossed indices is not an involution.

Let $a$ have rank $n-1$. Its noncyclic states form a unique tail
$T_0,\ldots,T_{t-1}$ entering a cycle $C_0,\ldots,C_{c-1}$, where
\[
 a(T_i)=T_{i+1}\ (i<t-1),\qquad a(T_{t-1})=C_0,
 \qquad a(C_j)=C_{j+1\bmod c}.
\]
Other $a$-cycles may be present. Assume $t\le c$ and put $v=c-t$.
For an involution $b$, let $I$ be the set of tail indices exchanged with
cyclic states. We say that its TC support is \emph{exactly aligned} if
\begin{equation}\label{S-eq:full-aligned-support}
 I\subseteq\{0,\ldots,t-1\},\qquad
 X=\{C_{v+i}:i\in I\},\qquad b(T_i)=C_{v+f(i)}\quad(i\in I),
\end{equation}
where $X$ is the entire set of cyclic endpoints of TC transpositions and
$f$ is a permutation of $I$. Define
$\operatorname{span}(I)=\max I-\min I$ for nonempty $I$.

\begin{Stheorem}\label{S-thm:aligned-permutations}
Suppose $\operatorname{rank}(a)=n-1$, $b$ is an involution, $t\le c$,
and the entire TC support satisfies \eqref{S-eq:full-aligned-support}.
If $f^2\ne\mathrm{id}$, then
\[
 \Tr\le c+\operatorname{span}(I)+3\le n+2.
\]
A witnessing word can be chosen with at most two occurrences of $b$.
There is no restriction on TT or RR transpositions or on additional
$a$-cycles. Strong connectivity and synchronization are not required.
\end{Stheorem}

We prove the theorem by separating the inverse-word construction from a
finite-permutation lemma. All lengths below count original binary letters.

\subsection{A word with two occurrences of the involution letter}
Write $A(S)=a^{-1}(S)$ and $B(S)=b^{-1}(S)=b(S)$. Inverse instructions
are listed in their application order; reversing their letters gives the
ordinary forward word.

\begin{Slemma}\label{S-lem:span-two-B}
Under \eqref{S-eq:full-aligned-support}, suppose $J,e\in I$ satisfy
\begin{equation}\label{S-eq:span-certificate}
 e<J,\qquad \beta=e+f^2(J)-J\notin I,\qquad
 K=f(e)-e-f(J)+J\le\operatorname{span}(I).
\end{equation}
Then a word of length at most $c+\operatorname{span}(I)+3$ merges
some triple and uses at most two $b$ letters.
\end{Slemma}

\begin{proof}
Starting from $\{C_0\}$, use the inverse instructions
\begin{equation}\label{S-eq:span-two-B-word}
 A^{t-f(J)},\ B,\ A^{J-e},\ B,\ A^{v+f(e)+1},
\end{equation}
stopping at the first set with at least three elements. The first block
gives $\{T_{f(J)},C_{v+f(J)}\}$, and the first $B$ gives
$\{T_J,C_{v+f^2(J)}\}$, by alignment and involutivity.

Put $q=v+f^2(J)$. If $q<J-e$, the cyclic coordinate reaches $C_0$
while its tail companion $T_{J-q}$ still has positive index. One more
$A$ then gives at least three preimages. The total cost is
\[
 t-f(J)+1+q+1=c+f^2(J)-f(J)+2
 \le c+\operatorname{span}(I)+2.
\]
Otherwise the middle block gives $\{T_e,C_{v+\beta}\}$. Here
$v+\beta\ge0$, and $\beta<f^2(J)<t$ ensures $v+\beta<c$.
Since $\beta\notin I$, the cyclic point is not a TC endpoint. The
second $B$ therefore gives two distinct cyclic points, one equal to
$C_{v+f(e)}$. The other may lie on any $a$-cycle and continues to have
a cyclic predecessor at every step. After at most $v+f(e)+1$ further
$A$ steps the receiving-cycle point contributes two preimages, giving
a triple. This costs
\[
 t-f(J)+1+J-e+1+v+f(e)+1=c+K+3.
\]
At $q=J-e$ the second $B$ is taken before any expansion is claimed;
this includes the pair $\{T_0,C_0\}$. Thus the boundary case is covered
as well. Reversing the truncated instruction sequence proves the claim.
\end{proof}

\subsection{A span bound for permutations with a cycle of length at least three}

\begin{Slemma}\label{S-lem:permutation-span}
Let $I$ be a finite set of integers and let $f$ be a permutation of $I$
with $f^2\ne\mathrm{id}$. There exist $J,e\in I$ satisfying
\eqref{S-eq:span-certificate}. The starting indices needed to find such a
pair can be selected from one cycle of $f$ of length at least three.
\end{Slemma}

\begin{proof}
Translating every index by $\min I$ preserves $K$, order, and membership
of $\beta$ after the same translation. Normalize $\min I=0$ and write
$M=\max I=\operatorname{span}(I)$. Suppose, for a contradiction, that
no pair satisfies \eqref{S-eq:span-certificate}.

\smallskip\noindent\emph{Residue confinement.}
For any start $J>f(J)>f^2(J)$, put
\[
 \alpha=J-f(J)>0,\qquad d=J-f^2(J)>\alpha.
\]
Every exit $e\in I$ with $\alpha\le e<J$ is within budget, since
\begin{equation}\label{S-eq:span-negative-budget}
 K=f(e)-e+\alpha\le M-e+\alpha\le M.
\end{equation}
In particular, for $d\le e<J$, absence of a certificate forces
$e-d\in I$. Repeated subtraction reaches the remainder $r=e\bmod d$;
all preceding inputs are at least $d>\alpha$, so they remain within
budget. If $r\ge\alpha$, then $r$ is itself a within-budget exit with
$r-d<0$ outside $I$, a contradiction. Therefore
\begin{equation}\label{S-eq:span-residue}
 e\bmod d\in\{0,\ldots,\alpha-1\}
 \qquad(e\in I,\ e<J).
\end{equation}
This also applies to an initial exit smaller than $d$.

\smallskip\noindent\emph{Choice of a descending segment.}
Select any cycle $\mathcal C$ of $f$ of length at least three. It has an
ascent. Choose its largest ascent origin:
\[
 Q=\max\{q\in\mathcal C:f(q)>q\},\qquad D=f(Q)>Q.
\]
Every $x\in\mathcal C$ with $x>Q$ satisfies $f(x)<x$: it cannot be
an ascent by the choice of $Q$, and it cannot be fixed on this long
cycle. Thus iterating $f$ from $D$ strictly decreases until the first
point $H\le Q$. Such a point exists because a strictly decreasing
sequence in a finite set cannot continue indefinitely. The point just
before $H$ is greater than $Q$. We distinguish the number of descending
edges from $D$ to $H$.

\smallskip\noindent\emph{One descending edge.}
Here $f(D)=H\le Q$. Equality $H=Q$ would make $Q,D$ a two-cycle,
so $H<Q$. At start $J=Q$, choose $e=0$. Then
\[
 \beta=H-Q<0,\qquad K=f(0)-(D-Q)<M,
\]
contrary to the assumed absence of a certificate.

\smallskip\noindent\emph{Two descending edges.}
Write $D\mapsto C\mapsto H$, so $H\le Q<C<D$, and put
\[
 \alpha=D-C,\quad b=C-H,\quad d=D-H=\alpha+b,
 \quad\delta=C-Q\in(0,b].
\]
At start $D$, the exit $C$ has $K=\alpha-b\le M$ and
$\beta=H-\alpha$. If it is not a certificate, then $H\ge\alpha$.
At the positive start $Q$, every exit $e<Q$ is within budget because
\[
 K=f(e)-e-(D-Q)\le M-e-(D-Q)<M.
\]
Starting at $0\in I$, add $\delta$ repeatedly as a selection procedure
until first reaching $e\ge\alpha$. Each preceding value is
$<\alpha\le H\le Q$, so absence of a certificate at $Q$ forces
the next value to belong to $I$. The first such value satisfies
\[
 \alpha\le e\le\alpha+\delta-1\le d-1<D.
\]
At start $D$ it is within budget by \eqref{S-eq:span-negative-budget},
and its candidate $e-d$ is negative. This is the required contradiction.

\smallskip\noindent\emph{At least three descending edges.}
Reverse the segment's notation:
\[
 H=y_0<y_1<\cdots<y_k=D,\quad k\ge3,\qquad
 f(y_i)=y_{i-1},\quad H\le Q<y_1.
\]
Put $g_i=y_i-y_{i-1}>0$. At start $y_i$, for $i\ge2$, the
residues of $y_{i-2},y_{i-1}$ modulo $g_{i-1}+g_i$ both lie in
$[0,g_i-1]$ by \eqref{S-eq:span-residue}. Their ordinary difference is
either $g_{i-1}$ or $-g_i$. The latter cannot fit in that interval,
and the former forces
\begin{equation}\label{S-eq:span-increasing-gaps}
 g_1<g_2<\cdots<g_k.
\end{equation}
Set
\[
 B=y_{k-2},\qquad C=y_{k-1},\qquad
 X=B-H,\quad Y=C-B,\quad Z=D-C.
\]
At start $C$, every prefix point $H,y_1,\ldots,B$ has residue in
$[0,Y-1]$ modulo $L=g_{k-2}+Y$. Each successive prefix gap $g$ is
at most $g_{k-2}$ by \eqref{S-eq:span-increasing-gaps}. A wrap in
residues would give difference
\[
 g-L\le -Y,
\]
impossible between two points of $[0,Y-1]$. Consequently no successive
step wraps, regardless of the number of prefix points. The total
residue spread is the true sum $X$, and hence
\begin{equation}\label{S-eq:span-XYZ}
 0<X<Y<Z.
\end{equation}
At start $D$, the exit $B$ has
$K=f(B)-B+Z<Z\le M$. Its candidate $B-(Y+Z)$ must therefore
belong to $I$ and be nonnegative. It follows that
\begin{equation}\label{S-eq:span-H-large}
 H\ge Y+Z-X>Z.
\end{equation}
Define
\[
 X'=B-Q,\qquad d=Y+Z,\qquad\delta=C-Q=X'+Y.
\]
The segment's construction gives
$0<X'\le X<Y$, since $H\le Q<y_1\le B$.
In particular $0<\delta<d$. At start $D$, residue confinement gives
the nonempty set
\begin{equation}\label{S-eq:span-R}
 R=I\cap\{0,\ldots,d-1\}\subseteq\{0,\ldots,Z-1\},
 \qquad 0\in R.
\end{equation}
Indeed $d=D-B<D$ because $B>0$. Both $Q$ and $B$ belong to $I$ and are
less than $D$. Repeatedly subtracting $d$ from them, with every input
at least $d>Z$ within budget, puts both remainders in $R$.

For $e\in R$, equations \eqref{S-eq:span-H-large}--\eqref{S-eq:span-R}
give $e<Z<H\le Q$. The positive start $Q$ is within budget and
forces $e+\delta\in I$. Moreover
\[
 e+\delta<Z+X'+Y=D-Q<D.
\]
Reduction modulo $d$ at start $D$, if needed, stays crossed and
within budget. Thus $R$ is closed under addition of $\delta$ modulo
$d$. As a finite set closed under a permutation, it is a union of
complete translation orbits.

Let $g=\gcd(d,\delta)$. These orbits are the residue classes modulo
$g$ inside $\{0,\ldots,d-1\}$. If $g\le Y$, every such class meets
the forbidden interval $\{Z,\ldots,d-1\}$ of $Y$ consecutive integers.
This contradicts \eqref{S-eq:span-R}.

If $g>Y$, then $g\mid\delta=X'+Y<2Y$ forces $g=\delta$, since
every proper positive divisor of $\delta$ is at most $\delta/2<Y$.
Thus $\delta\mid d$. Modulo $\delta$, the forbidden interval has
exactly the residues $\{X',\ldots,\delta-1\}$, because
$Z=d-Y\equiv X'\pmod\delta$. Every orbit contained in $R$ must
therefore have residue in $\{0,\ldots,X'-1\}$. But $R$ contains the
remainders of $Q$ and $B=Q+X'$ modulo $d$, and $\delta\mid d$.
Their residues modulo $\delta$ would both have to lie in this short
interval. Their ordinary difference is either $X'$ or $-Y$, neither
of which fits in an interval of width $X'-1$, since $Y>X'$.
This is the final contradiction.

All cases produce a certificate. Translating the selected indices back
to the original set preserves \eqref{S-eq:span-certificate}, completing
the proof.
\end{proof}

\begin{proof}[Proof of Theorem~\ref{S-thm:aligned-permutations}]
Lemmas~\ref{S-lem:permutation-span} and~\ref{S-lem:span-two-B}, applied to
$I$, give a word of length at most
\[
 c+\operatorname{span}(I)+3\le c+t+2\le n+2.
\]
The inequalities use $\operatorname{span}(I)\le t-1$ and $t+c\le n$.
\end{proof}

\subsection{Construction and remaining cases}
The proof is constructive. Select any $f$-cycle of length at least
three, choose its largest ascent origin $Q$, and record the descending
segment from $f(Q)$ to the first $H\le Q$. Scan the starting index $Q$
and every segment point having at least two subsequent points in the
segment, and for each scan $e\in I$ with $e<J$. These are all the
starts used in the proof, so one of at most $|I|^2$ pairs satisfies
\eqref{S-eq:span-certificate}. Its actual word is obtained from
\eqref{S-eq:span-two-B-word}, using the original physical indices.
Arithmetic progressions and modular orbits select the pair; they do
not add letters to that word.

If $t>c$, the independent word $a^{c+1}$ already merges a triple,
without alignment. For $t\le c$, alignment must describe the entire
TC support: additional TC endpoints on other cycles would invalidate
the second $B$ step. When $f^2=\mathrm{id}$, every proposed
$\beta=e+f^2(J)-J$ equals $e\in I$, so this certificate cannot apply.
The general aligned case $f^2=\mathrm{id}$, unaligned support, and the
unrestricted rank-$(n-1)$/involution bound remain open beyond the
separate cases proved earlier.

\section{Conjectures and questions}\label{S-sec:conj}

\begin{Squestion}\label{S-q:general-binary}
Is $\Tr\le 4n/3$ for every synchronizing binary automaton with $n\ge6$ states, or at least for
those with a permutation letter and a letter of rank $n-1$? (For $n=5$ the maximum is $7>20/3$,
Table~\ref{S-tab:max}.)
\end{Squestion}

\begin{Squestion}\label{S-q:coincidence}
Is there a structural reason for the coincidence $\max\Tr(n)=\CWA(n)=\floor{4n/3}$
(Main Definition~\ref{M-def:cwa}) observed for $8\le n\le 10$, or do the two maxima eventually
diverge? Main Corollary~\ref{M-cor:tau} shows that for
$n\not\equiv0\pmod 3$ the automaton $\G{r}{k}$ itself has a state of threshold $\floor{4n/3}-1$.
\end{Squestion}
The binary larger-subset question and its finite ranges appear in Section~\ref{S-sec:largerk}.
\appendix
\section{The three-transposition case proofs}\label{S-app:three-proofs}
The following lemmas provide the full case arguments used in
Theorem~\ref{S-thm:three-swaps}. Inverse operators are read in application
order; all lengths count original letters.
\subsection{A tail of length one}

\begin{Slemma}\label{S-lem:three-tail-one}
Let $\mathcal A=(Q,\{a,b\})$ be a synchronizing automaton on $n$ states,
where $a$ has rank $n-1$ and $b$ is a product of exactly three disjoint
transpositions. Suppose that $a$ restricts to a permutation on
$R=Q\setminus\{h\}$, where $h$ is its missing image. Then
$\Tr(\mathcal A)\le n+2$. Strong connectivity is not required.
\end{Slemma}

\begin{proof}
Write $\pi=a|_R$, $z=a(h)$, $p=\pi^{-1}(z)$, and $K=\{h,p\}$.
Let $C$ be the $\pi$-cycle containing $p,z$, of length $c$. We use full
preimages $A(S)=a^{-1}(S)$ and $B(S)=b(S)$ in application order, beginning
at $\{z\}$. Reversing this order gives a forward word. For a pair $P$,
\[
 |A(P)|=2+\mathbf1_{z\in P}-\mathbf1_{h\in P}.
\]
Thus a pair in $R$ containing $z$ expands to three under $A$; before its
first visit to $z$, $A$ rotates a pair in $R$ by $\pi^{-1}$. From a
cyclic pair touching $C$, at most $c$ further $A$-steps suffice, including
the final expansion. All bounds below reserve these $c$ steps, whether
or not fewer actually suffice.

For later exclusions, a shortest successful inverse trajectory can start
with $A(\{z\})=K$ and remain a pair until its first expansion to three.
After a singleton contraction it could instead restart at the next
singleton-to-pair expansion, which again gives $K$, since $z$ is the
only double image. Consequently a closed family of pairs, allowing
contractions to singletons but containing no expansion target, excludes
synchronization.

The state $h$ must be moved by $b$: otherwise both letters permute the
invariant set $R$. Put $x=b(h)$. If $x=p$, then $B(K)=K$, whereas
$A(K)$ is $K$ when $c=1$ and a singleton otherwise. The preceding
criterion excludes this case. Hence $x\ne p$, and
$F=B(K)=\{x,b(p)\}\subset R$. If $F$ touches $C$, the initial $A,B$
and the reserved final $c$ steps give $\Tr\le c+2\le n+1$.

Assume therefore that $F$ avoids $C$. Necessarily
\[
 b=(h\ x)(p\ y)(u\ v),\qquad x,y\notin C,
\]
with all six states distinct. Form the graph on the $\pi$-cycles whose
two edges are $(p,y)$ and $(u,v)$; loops are permitted. This graph is
connected. Indeed a component not containing $x$ is invariant and
permuted by both letters. If it has at least two states, synchronization is
impossible. If it is a singleton, it lies outside $C$, since $(p,y)$
joins $C$ to another cycle; that singleton and its complement are both
invariant, again a contradiction. Hence there are two or three cycles.

\smallskip\noindent
\emph{A non-antipodal pair on one external cycle.}
Suppose $x,y$ belong to a cycle $D$ of length $d$. Use coordinates
$D=\mathbb Z_d$ with $\pi(i)=i+1$, $x=0$, and $y=\Delta$.
If $2\Delta\not\equiv0\pmod d$, rotating $F$ by $j=d-\Delta$
inverse steps gives $\{y,w\}$ with $w\ne x$. The next $B$ sends $y$
to $p$ and does not send the other point to $h$. The resulting cyclic
pair touches $C$, regardless of the location of $(u,v)$. The cost is
\[
 1+1+j+1+c\le c+d+2\le n+1.
\]
In all remaining cases with $x,y\in D$, write $d=2m$ and $y=m$;
then $x,y$ are antipodal. Write $\sigma(i)=i+m$.

\smallskip\noindent
\emph{Two cycles $C,D$.}
Here $n=1+c+d$. If $(u,v)$ is internal to $C$, the antipodal pairs of
$D$ are closed under $A$, and $B$ fixes them except that $\{x,y\}$
returns to $K$. This is a closed family without a target, impossible.
The same obstruction applies when $u,v\in D$ are antipodal.

If the extra swap crosses from $u\in D$ to $v\in C$, rotate until
the original $y$-coordinate reaches $u$. This takes $1\le j\le d-1$
steps. The other coordinate is not $x$, since the rotation is nonzero.
Thus the next $B$ produces a cyclic pair touching $C$, at total cost
$j+c+3\le c+d+2$.

It remains that $u,v\in D$ and $v\ne\sigma(u)$. If an endpoint lies
in $\{1,\ldots,m-1\}$, name it $u$ and set $j=m-u$, $k=u$.
After $j$ rotations the pair is $\{u,u+m\}$; $B$ gives
$\{v,u+m\}$. After $k$ more rotations the unchanged point is $y$,
and its companion is not $x$, since that would require $v=u$.
Here $j+k=m$. If neither endpoint lies in that interval, both lie in
$\{m+1,\ldots,d-1\}$. Put $j=d-u$ and $k=v-m$. The same two
stages give $\{u,u-m\}$ and then $\{v,u-m\}$; after $k$ rotations
the moved endpoint is $y$ and its companion is $u-v\ne0$ modulo $d$.
Now $j+k=m+v-u\le d-2$. In either case a further $B$ reaches a cyclic
pair containing $p$, and the reserved completion gives
\[
 1+1+j+1+k+1+c\le c+d+3=n+2.
\]

\smallskip\noindent
\emph{Three cycles, with $x,y\in D$.}
Let the third cycle be $E$ of length $e$, so $n=1+c+d+e$. The extra
swap joins $E$ to $C$ or $D$. A $C$--$E$ swap leaves the antipodal
family in $D$ closed as above and is excluded. Otherwise write
$u\in D$, $v\in E$. Rotate until one of $x,y$ first reaches $u$;
the cost $j$ satisfies $1\le j\le m-1$. The other point is
$w=\sigma(u)$, neither $x$ nor $y$. Applying $B$ gives $\{v,w\}$.
Rotate by $k$ steps until the $D$-point is $y$, then apply $B$ again.
The $E$-point stays on $E$ or moves to $u\in D$, so the resulting
cyclic pair contains $p$. The untouched $D$-point started at $x$ or
$y$, giving respectively $j+k=m$ or $j+k=d$. Hence the total is at most
\[
 c+j+k+4\le c+d+4\le c+d+e+3=n+2.
\]

\smallskip\noindent
\emph{Three cycles, with $x\in E$ and $y\in D$.}
Write $|E|=e$, $|D|=d$. Connectivity forces the extra swap to join
$E$ to $C$ or $D$. In the first case, name $u\in E$, $v\in C$.
Rotating $x$ to $u$ takes at most $e-1$ steps; the other coordinate
stays in $D$. One $B$ reaches a cyclic pair touching $C$, for cost
at most $c+e+2$.

In the second case, name $u\in E$, $v\in D$, and let
$\alpha,\beta$ be the inverse rotation distances from $x$ to $u$
and from $y$ to $v$. Thus
$1\le\alpha<e$, $1\le\beta<d$. If $e\nmid d$, rotating by $d$
returns the second coordinate to $y$ but not the first to $x$.
One $B$ and the final completion cost at most $c+d+3\le n+2$.

Suppose $e\mid d$. If $\alpha\ne\beta$, put $j=\alpha$; if
$\alpha=\beta$ and $e<d$, put $j=\alpha+e$. In either case $j<d$,
the first point is $u$, and the second is not $v$. Applying $B$ puts
both points in $D$. After $d-j$ further rotations the original
$D$-point is $y$; another $B$ gives a cyclic pair containing $p$.
The cost is
\[
 1+1+j+1+(d-j)+1+c=c+d+4\le n+2,
\]
since $e\ge2$ here. The remaining case is $e=d$, $\alpha=\beta$.
The pairs
\[
 \{\pi^{-i}(x),\pi^{-i}(y)\},\qquad i\in\mathbb Z_d,
\]
form a matching preserved by $A$. The map $B$ fixes these pairs except
that $\{x,y\}$ returns to $K$; at $i=\alpha$ it merely exchanges
$u,v$ inside one pair. This closed family is impossible in a
synchronizing automaton. All cases are covered.
\end{proof}

\subsection{Only the missing image moves on the tail}

\begin{Slemma}\label{S-lem:three-one-tail}
Let $\mathcal A=(Q,\{a,b\})$ be synchronizing, with $|Q|=n$,
$\operatorname{rank}(a)=n-1$, and $b$ a product of exactly three
disjoint transpositions. If the only noncyclic $a$-state moved by $b$
is its missing image $h$, then $\Tr(\mathcal A)\le n+2$.
Strong connectivity is not required.
\end{Slemma}

\begin{proof}
Write the unique tail and receiving cycle as
\[
 T_0=h\longrightarrow T_1\longrightarrow\cdots\longrightarrow
 T_{t-1}\longrightarrow C_0=z\longrightarrow\cdots\longrightarrow
 C_{c-1}\longrightarrow C_0.
\]
If $t>c$, the word $a^{c+1}$ has a triple fibre. Assume $t\le c$,
put $v=c-t$, $p=C_v$, and retain the inverse notation $A,B$ above.
The initial pair chain is
\[
 P_i=A^{t-i}(\{z\})=\{T_i,C_{v+i}\},\qquad 0\le i<t.
\]
No $P_i$ is an expansion target. At $P_0$, $A$ contracts to a
singleton if $v>0$ and returns to the first collision pair if $v=0$.
The closed-pair criterion from Lemma~\ref{S-lem:three-tail-one} therefore
applies with the whole initial chain included.

First suppose $b(h)=C_x$. If $x\ne v$, $B(P_0)$ is a cyclic pair
touching $C$, and the cost is at most $t+1+c\le n+1$. If $x=v$,
$B$ fixes $P_0$. Unless a cyclic transposition touches some
$C_{v+i}$ with $1\le i<t$, it fixes every other $P_i$ too, giving
an impossible closed chain. Choose such a point and call its partner
$U$. At cost $t-i+1$ we reach $\{T_i,U\}$. If the cyclic point
first reaches $z$ in $k<i$ inverse rotations, one more $A$ expands
while the tail point is positive, at total cost at most $t+1$.
Otherwise $i$ rotations give $\{h,U'\}$ with $U'\ne p$: this is
automatic off $C$, and if $U=C_k$ then $k\ge i$ and equality
$C_{k-i}=p$ would require $k=v+i$, contradicting the transposition.
One $B$ now gives a cyclic pair including $p$. At most $v+1$ final
steps suffice, for total
\[
 (t-i)+1+i+1+v+1=c+3\le n+2.
\]
Both other transpositions may be used in this reasoning; only $p$ can
be sent to $h$.

Now suppose $x=b(h)$ lies outside $C$. If $B(P_0)$ touches $C$, the
same $t+c+1$ construction applies. Otherwise
\[
 b=(h\ x)(p\ y)(u\ w),\qquad x,y\notin C.
\]
The graph on cyclic components with edges $(p,y),(u,w)$ is connected
by the invariant-component argument in Lemma~\ref{S-lem:three-tail-one};
the singleton exception is again outside $C$. Thus there are two or
three cyclic components, and its complete pair-route classification
applies. This does not assume that an auxiliary automaton synchronizes.
One may use the tail-one map agreeing with $a$ on cyclic states and
sending $h$ to $a(p)$ only to identify its initial pair as $P_0$.

Every constructive branch in that lemma starts at $P_0$ with $B$ and
stays on cyclic pairs until it touches $C$. If $q$ is the length of
this prefix, its proof reserves $c$ final steps and establishes
$1+q+c\le n'+2$, where $n'=1+\#\{\text{cyclic states}\}=n-t+1$.
In the present automaton the initial cost is $t$, and the actual
target $z$ still permits completion within $c$ steps. An earlier
visit to $z$ only shortens the route. Consequently
\[
 t+q+c\le n'+2+t-1=n+2.
\]
This transfers the proved budget, not the actual length of an
auxiliary word with a different terminal target.

It remains to inspect every closed family in the classification,
since the extra initial pairs could provide an escape. There are
precisely four possibilities: $x,y$ are antipodal on one external
cycle $D$ and the extra swap is internal to $C$, antipodal within
$D$, or joins $C$ to a third cycle $E$; or $x,y$ lie on equal-length
external cycles and the extra bridge joins corresponding matching
points. In the internal-$D$ and equal-matching cases, the only moved
$C$-point is $p$. Every $P_i$ with $i\ge1$ is fixed by $B$, so the
old matching together with the chain remains closed and is excluded.

If the extra swap is internal to $C$ and avoids all $C_{v+i}$,
$i\ge1$, the same closure holds. Otherwise it sends some $C_{v+i}$
to $C_k$, $k\ne v+i$. From $P_i$ one $B$ gives $\{T_i,C_k\}$.
For $k<i$ the positive-tail early expansion applies. For $k\ge i$,
another $i$ rotations give $\{h,C_{k-i}\}$ with $C_{k-i}\ne p$.
Its $B$-image is a cyclic pair with one point in $C$ and one in $D$.
The cost, including completion, is at most
\[
 (t-i)+1+i+1+c=t+c+2\le n+2.
\]

Finally suppose the extra swap is $(C_q,V)$ with $V\in E$.
If $q$ is not $v+i$ for any $i\ge1$, the initial chain is fixed and
the matching remains closed. Otherwise $B(P_i)=\{T_i,V\}$ and
$A^i(\{T_i,V\})=\{h,V'\}$ with $V'\in E$. Applying $B$ gives
$\{x,b(V')\}$. If its other point lies on $C$, completion costs at
most $t+c+2$ in total. Otherwise that point stays on $E$. Write
$|D|=2m$ and $|E|=e$. After $m$ inverse rotations the $D$-point is
$y$; one more $B$ sends it to $p$ and cannot send the $E$-point to
$h$. A final completion gives the bound
\[
 (t-i)+1+i+1+m+1+c=t+c+m+3
 \le t+c+2m+e+2=n+2,
\]
since $m,e\ge1$. Thus every closed family either remains impossible
or has an additional short escape. This proves the lemma.
\end{proof}

\subsection{Two moved tail points}

\begin{Slemma}\label{S-lem:three-two-tail}
Let $(Q,\{a,b\})$ be a strongly connected synchronizing automaton with
$|Q|=n$ and $\operatorname{rank}(a)=n-1$. Write its unique $a$-tail and
receiving cycle as
\[
T_0=h\longrightarrow T_1\longrightarrow\cdots\longrightarrow T_{t-1}
\longrightarrow C_0=z\longrightarrow C_1\longrightarrow\cdots
\longrightarrow C_{c-1}\longrightarrow C_0,
\]
where $h$ is the missing image and $z$ is the double image of $a$.
Suppose $2\le t\le c$ and
\[
b=(T_0,X)(T_j,Y)(U,V),\qquad 1\le j<t,
\]
where $X,Y,U,V$ are four distinct cyclic states of $a$, and all other
states are fixed by $b$. Additional $a$-cycles are allowed. Then
$\Tr\le n+2$.
\end{Slemma}

\begin{proof}
Put $v=c-t$, $A(S)=a^{-1}(S)$ and $B(S)=b(S)$. We list inverse operations
in their application order, starting at $\{z\}$; reversing the listed
letters gives the corresponding forward word. The initial pair chain is
\begin{equation}\label{S-eq:three-two-tail-chain}
P_i=A^{t-i}(\{z\})=\{T_i,C_{v+i}\},\qquad 0\le i<t.
\end{equation}
The first pair is $P_{t-1}$. For every pair $P$,
\[
|A(P)|=2+\boldsymbol{1}_{z\in P}-\boldsymbol{1}_{h\in P}.
\]
Call a pair containing $z$ but not $h$ a target. One further $A$ expands
a target to three states. At $P_0$, $A$ contracts to a singleton if $v>0$
and returns to $P_{t-1}$ if $v=0$.

Two observations will be used repeatedly. First, a pair of cyclic
states containing $C_q$ reaches a target in at most $q$ inverse $a$-steps:
stop when either cyclic coordinate first reaches $z$. One further $A$
expands it, so the completion costs at most $q+1\le c$.
Second, a family of nontarget pairs containing $P_{t-1}$ and closed under
$A,B$ except for contractions to singletons rules out synchronization.
Indeed, a shortest inverse trajectory from any singleton to three states
must first expand from $\{z\}$ to the unique collision pair $P_{t-1}$.
If it subsequently contracts to a singleton, its prefix can be discarded
and the same argument restarted. Thus a closed family of this kind
prevents every triple from being merged.

\smallskip\noindent
\emph{Reduction to two configurations.}
For $i\in\{0,j\}$, if $B(P_i)$ is a cyclic pair meeting the receiving
cycle $C$, the preceding completion gives length at most
\[
(t-i)+1+c\le t+c+1\le n+1.
\]
Assume neither choice works, and put $\alpha=C_v$, $\beta=C_{v+j}$.
If either point is outside $\{X,Y,U,V\}$, it is fixed by $b$ and gives
the construction just described. Hence both belong to this four-point
set. They cannot be $U,V$: in that case the cyclic transposition retains
a point on $C$ and again gives the construction. There remain exactly
two configurations:
\begin{enumerate}
\item $\{\alpha,\beta\}=\{X,Y\}$;
\item for $\{r,q\}=\{0,j\}$, after naming the external endpoints,
\begin{equation}\label{S-eq:three-two-tail-external-form}
b=(T_r,Z)(T_q,C_{v+q})(C_{v+r},V),\qquad Z,V\notin C.
\end{equation}
\end{enumerate}
For the second assertion, the initial pair whose cyclic coordinate is an
endpoint of the cyclic transposition is sent to $\{Z,V\}$. Neither point
can lie on $C$, by our assumption. The remaining cyclic support point on
$C$ is therefore the cross endpoint matched to the other tail port, as
in~\eqref{S-eq:three-two-tail-external-form}.

\smallskip\noindent
\emph{Both cross endpoints on the initial chain.}
Suppose $\{X,Y\}=\{C_v,C_{v+j}\}$. Then $B$ fixes or exchanges $P_0,P_j$,
and fixes every other $P_i$ unless $C_{v+i}\in\{U,V\}$. If neither
endpoint of $(U,V)$ belongs to
\[
\{C_{v+i}:1\le i<t,\ i\ne j\},
\]
the initial chain is closed and has no target, a contradiction. Choose
an endpoint $C_{v+i}$ in this set. If both endpoints lie on $C$, choose
the one with larger coordinate; it still belongs to the displayed set.
Writing its partner as $C_w$, we then have $w<v+i$.

The operations $A^{t-i},B$ produce
$\{T_i,b(C_{v+i})\}$, because $T_i$ is fixed by $b$. If the cyclic
partner is outside $C$, applying $A^i$ gives $\{h,D\}$ with $D$ still on
an external cycle. The next $B$ gives two cyclic states, including
$b(h)=X\in C$: the other state is fixed or returns to $C$ through the
cyclic transposition, and cannot be a cross endpoint. The total cost is
at most
\begin{equation}\label{S-eq:three-two-tail-hole-cost}
(t-i)+1+i+1+c=t+c+2\le n+2.
\end{equation}
If instead the partner is $C_w$ and $w<i$, applying $A^w$ reaches $z$
while the tail coordinate is still positive. The next $A$ expands, at
cost $(t-i)+1+w+1\le t+1$. If $w\ge i$, applying $A^i$ yields
$\{h,C_{w-i}\}$ with $0\le w-i<v$. This cyclic coordinate is neither
cross endpoint, both of which have coordinate at least $v$. Thus the
next $B$ gives a cyclic pair meeting $C$, even if it uses $(U,V)$, and
\eqref{S-eq:three-two-tail-hole-cost} applies. When $v=0$, the inequality
$0\le w-i<v$ is impossible, so the earlier expansion necessarily occurs.

\smallskip\noindent
\emph{One fixed initial pair and an external pair.}
Now assume~\eqref{S-eq:three-two-tail-external-form} and set $U=C_{v+r}$.
The operations $A^{t-r},B$ give $\{Z,V\}$; $B$ fixes $P_q$ and every
other initial pair. Strong connectivity excludes any additional
$a$-cycle containing neither $Z$ nor $V$, since such a cycle would be
invariant under both letters and disjoint from $C$.

Suppose an inverse rotation $A^k$ of $\{Z,V\}$ places one point at $V$
and the other at a point different from $Z$. The latter point is also
different from $V$, since rotation is injective on cyclic states, and
is fixed by $b$. The next $B$ gives a cyclic pair containing $U$.
Completing at $C_{v+r}$ gives total cost at most
\begin{equation}\label{S-eq:three-two-tail-rotation-cost}
(t-r)+1+k+1+(v+r)+1=c+k+3.
\end{equation}

If $Z,V$ lie on the same external cycle of length $d$, let
$1\le k\le d-1$ be the inverse rotation carrying $Z$ to $V$. Unless
the two points are antipodal, the other image is different from $Z$,
and~\eqref{S-eq:three-two-tail-rotation-cost} gives
$\Tr\le c+d+2\le n+2$, where $n=t+c+d$.
If they are antipodal, every rotation of the pair contains both marked
points $Z,V$ or neither. Hence $B$ returns it to $P_r$ or fixes it.
The rotations and the initial chain then form a closed family without
a target, contradicting synchronization.

It remains that $Z,V$ lie on distinct cycles of lengths $e,d$,
respectively; here $n=t+c+e+d$. If $e\nmid d$, take $k=d$: the
$V$-coordinate returns to $V$ while the $Z$-coordinate does not return
to $Z$. Thus~\eqref{S-eq:three-two-tail-rotation-cost} gives
$\Tr\le c+d+3\le n+2$.

Suppose $e\mid d$. If $r=j>0$ and $e<d$, rotate by $e$ instead, obtaining
$\{Z,V'\}$ with $V'\ne V$. The next $B$ gives $\{T_j,V'\}$, and $A^j$
gives $\{h,V''\}$, where $V''$ is still on the $V$-cycle. Another $B$
gives a cyclic pair containing $b(h)=C_v$. Its other point is cyclic
even if $V''=V$, since $Z$ lies on a different cycle. The resulting cost
is at most
\[
(t-j)+1+e+1+j+1+v+1=c+e+4\le n+2;
\]
indeed $t\ge2$ and $d>e\ge1$ give $t+d\ge4$.

The remaining divisibility cases cannot synchronize. If $e=d$, the two
coordinates return to their marked points simultaneously. Each pair
rotation is therefore fixed by $B$ or returned to $P_r$, so these
rotations and the initial chain form a closed family without a target.
If $r=0$ and $e\mid d$, any pair rotation containing $V$ also contains
$Z$, and $B$ returns it to $P_0$. A rotation containing $Z$ alone is
sent to $\{h,V'\}$, with $V'$ on the external $V$-cycle; $A$ contracts
this pair to a singleton and $B$ returns to the external pair. All
remaining rotations are fixed by $B$. Adding these $\{h,V'\}$ pairs
to the initial chain and the external rotations again gives a closed
family without a target. Both alternatives contradict synchronization.
All configurations have now been exhausted, proving the bound.
\end{proof}

\subsection{Three moved tail points and one cyclic exchange}

\begin{Slemma}\label{S-lem:three-three-tail}
Let $\mathcal A$ be strongly connected and synchronizing, with
$\operatorname{rank}(a)=n-1$. Suppose $b$ consists of one tail--cycle
transposition, one tail--tail transposition, and one cycle--cycle
transposition. Then $\Tr(\mathcal A)\le n+2$.
\end{Slemma}

\begin{proof}
Use the tail $T_0=h,\ldots,T_{t-1}$, receiving cycle
$C=(C_0=z,\ldots,C_{c-1})$, and inverse operators $A=a^{-1}$, $B=b$.
If $t>c$, the word $a^{c+1}$ merges three states, so assume
$3\le t\le c$ and put $v=c-t$.
The initial pairs are $P_s=A^{t-s}(\{z\})=\{T_s,C_{v+s}\}$.
Their $A$ edges stay on the chain, apart from a singleton contraction
when $v>0$ or a return to the collision pair when $v=0$.
Closure of this entire chain under $B$ would therefore exclude every
triple. A cyclic pair touching $C$ expands after its first visit to $z$,
followed by one further $A$.

Write the cross edge as $(T_i,X)$. Strong connectivity forces $h$ to
move. Thus either $i=0$ and the tail exchange is $(T_p,T_q)$ with
$1\le p<q<t$, or it is $(h,T_j)$ with $0<i<j<t$.
For the latter inequality, a path from a cycle first enters the tail at
$T_i$ and must reach $T_j$ before using its edge to $h$; if $i>j$,
there is no first entry into the earlier tail segment.

First let $X=C_x$. If $x\ne v+i$, then $B(P_i)$ is a cyclic pair
containing $C_x$, giving total cost at most $(t-i)+1+c\le n+1$.
The cyclic exchange may move the other endpoint, but cannot send it
to the tail. Assume $x=v+i$.

If $i=0$, then $B(P_0)=P_0$. If every $B(P_s)$ stayed in the initial
chain, there would be no triple. Hence for some $s\ge1$,
$B(P_s)=\{T_j,W\}\ne P_j$, with $W$ cyclic.
If $W$ first reaches $z$ after $q<j$ inverse steps, one further $A$
expands, at cost at most $t+j-s+1\le2t-1\le n$.
Otherwise $A^j$ gives $\{h,W'\}$ with $W'\ne C_v$.
Indeed this is automatic on another cycle; if $W=C_k$, absence of an
earlier $z$ hit gives $k\ge j$, and $k-j=v$ would mean
$B(P_s)=P_j$. A further $B$ now gives two cyclic points, including
$C_v$. The first $z$ visit takes at most $v$ steps, and the total cost
through expansion is
\[
 (t-s)+1+j+1+v+1=c+3+j-s\le t+c+1\le n+1.
\]
The equality $q=j$ is handled by this latter route, not by an expansion
at a possible $\{h,z\}$.

If $i>0$, let the tail exchange be $(h,T_j)$, $i<j$, and put $d=j-i$.
From $P_0$, $B$ gives $\{T_j,W\}$ with $W=B(C_v)$ cyclic.
If $W$ reaches $z$ after $q\le d$ steps, the tail index is at least
$i\ge1$, so one more $A$ expands at cost $t+q+2\le2t\le n$.
Otherwise $A^d$ gives $\{T_i,W'\}$. If $W'\ne X$, the next $B$
gives a cyclic pair containing $X=C_{v+i}$, for cost at most
\[
 t+1+(j-i)+1+(v+i)+1=c+j+3\le t+c+2\le n+2.
\]
If $W'=X$, then $W$ must be $C_{v+j}$: on $C$ no wrap occurred
before the first $z$ hit, and on another cycle equality is impossible.
The cyclic exchange is therefore exactly $(C_v,C_{v+j})$.
Now $B$ exchanges $P_0,P_j$, fixes $P_i$ by swapping its endpoints,
and fixes every other initial pair. This closed chain contradicts
synchronization. Thus the receiving-cycle case is complete.

Next let $X$ lie on a cycle $D$ outside $C$.
The cyclic exchange must be $(C_y,Y)$ with $Y\in D$:
if it misses $C$, then $C$ is closed; if it joins $C$ to a different
cycle $E$, then $C\cup E$ is closed. Neither can reach $h$.
There are no further cycles, so $n=t+c+d$, where $d=|D|\ge2$.

For $i=0$, if $y\ne v$, then $B(P_0)=\{X,C_v\}$ expands at total
cost $c+2$. If $y=v$, write the tail exchange as $(T_p,T_q)$, $p<q$.
The inverse sequence $A^{t-q},B,A^p,B$ gives
\[
 \{T_q,C_{v+q}\},\quad\{T_p,C_{v+q}\},\quad
 \{h,C_{v+q-p}\},\quad\{X,C_{v+q-p}\}.
\]
Both positive-offset $C$ coordinates avoid the sole moved $C$ point
$C_v$. Applying $A^{v+q-p}$ and expanding costs
$(t-q)+1+p+1+(v+q-p)+1=c+3$.

For $i>0$, write the tail exchange as $(h,T_j)$ with $i<j$.
If $y\ne v+i$, $B(P_i)=\{X,C_{v+i}\}$ gives cost $c+2$.
If $y=v+i$, then $B(P_0)=\{T_j,C_v\}$.
When $v\le j-i$, applying $A^v$ reaches
$\{T_{j-v},z\}$ with $j-v\ge i\ge1$; the expansion costs $c+2$.
Otherwise $A^{j-i}$ gives $\{T_i,C_w\}$ with
$w=v-j+i>0$. Since $y-w=j>0$, $C_w$ is fixed by $B$.
The next $B$ gives $\{X,C_w\}$, and $A^w$ followed by an expansion
costs $t+1+(j-i)+1+w+1=c+3$.

Thus every outside-cycle case has cost at most $c+3\le n-2$;
the receiving-cycle cases have cost at most $n+2$. All costs count
the original binary letters, and all stated first-hit rules avoid
applying a further inverse step after an earlier expansion.
\end{proof}

\subsection*{Data and code availability}
The accompanying data contain the census programs, their outputs and
reproducible commands. Section~\ref{S-sec:comp} specifies which enumerations
are complete and which cover restricted classes.

\subsection*{Use of generative AI}
Generative AI tools, including OpenAI Codex, assisted with mathematical
analysis, research programming, literature review and manuscript preparation.
The author takes responsibility for the contents.

\end{document}